\documentclass{svjour3-edited}
\usepackage{amssymb,amsfonts,amsmath,amsbsy, pstricks,bbm,mathrsfs, multirow, bigstrut,graphicx,fixmath, enumerate}
\usepackage[english]{babel} 
\usepackage{bm}
\usepackage[ruled,vlined]{algorithm2e}
\newtheorem{observation}{Observation}
\newcommand{\universalgates}{\mathcal{U}}
\newcommand{\nodes}{\mathit{nodes}}
\newcommand{\arcs}{\mathit{arcs}}

\newcommand{\N}{{\mathbb{N}}}
\newcommand{\ket}[1]{|#1\rangle}
\newcommand{\bra}[1]{\langle #1 |}

\newcommand{\graph}{G}
\newcommand{\graphcal}{\mathcal{G}}

\newcommand{\treewidth}{\mathbf{tw}}

\newcommand{\height}{\mathit{height}}

\newcommand{\valuetensornetwork}{\mathrm{val}}
\newcommand{\valuefeasibilitytensornetwork}{\mathrm{VAL}}
\newcommand{\simulation}{\hat{\lambda}}
\newcommand{\evaluation}{\hat{\lambda}}

\newcommand{\truncatedfeasibilitysimulation}{\hat{\Lambda}}
\newcommand{\cutwidth}{\mathbf{cw}}
\newcommand{\onlinecutwidth}{\mathbf{ow}}
\newcommand{\poly}{\mathit{poly}}

\newcommand{\inputvertices}{\mathit{In}}
\newcommand{\outputvertices}{\mathit{Out}}
\newcommand{\internalvertices}{\mathit{Mid}}
\newcommand{\R}{{\mathbb{R}}}
\newcommand{\atensor}{g}
\newcommand{\tensorset}{\mathcal{F}}

\newcommand{\abstractnetwork}{\mathcal{N}}
\newcommand{\indexset}{\mathcal{I}}
\newcommand{\rank}{\mathit{rank}}
\newcommand{\tensors}{\mathbb{T}}

\newcommand{\carvingwidth}{\mathit{carw}}
\newcommand{\leaves}{\mathit{leaves}}
\newcommand{\width}{\mathbf{w}}

\newcommand{\vertexlabelingfunction}{{\theta}}
\newcommand{\edgelabelingfunction}{{\xi}}

\newcommand{\tensorContraction}{\mathit{Contr}}

\newcommand{\truncation}{{\mathit{Trunc}}}

\newcommand{\tensorSpace}{\mathbb{T}}

\newcommand{\feasibilitySet}{\mathcal{F}}
\newcommand{\operators}{\bm{L}}
\newcommand{\hilbert}{\mathcal{H}}
\newcommand{\trace}{{\mathit{tr}}}
\newcommand{\lang}{\mathcal{L}}
\newcommand{\C}{{\mathbb{C}}} 
\begin{document}

\title{On the Satisfiability of Quantum Circuits of Small 
Treewidth \footnote{This is an extended version of a paper that appeared at CSR 2015 \cite{deOliveiraOliveira2015Satisfiability}.}
\thanks{
This work was supported by the European Research Council, ERC grant agreement 339691, within the context 
of the project Feasibility, Logic and Randomness (FEALORA).}
}
\author{Mateus de Oliveira Oliveira}
\institute{M. de Oliveira Oliveira  \at 
Institute of Mathematics - Czech Academy of Sciences \\ 
\v{Z}itn\'{a} 25,  CZ - 115 67 Praha 1, Czech Republic \\
\email{mateus.oliveira@math.cas.cz} }

\maketitle
\begin{abstract}
It has been known for almost three decades that many $\mathrm{NP}$-hard optimization problems 
can be solved in polynomial time when restricted to structures of constant treewidth. 
In this work we provide the first extension of such results to the quantum setting. 
We show that given a quantum circuit $C$ with $n$ uninitialized inputs, $\mathit{poly}(n)$ gates,
and treewidth $t$, one can compute in time $(\frac{n}{\delta})^{\exp(O(t))}$ 
a classical assignment $y\in \{0,1\}^n$ that maximizes the acceptance 
probability of $C$ up to a $\delta$ additive factor. 
In particular, our algorithm runs in polynomial time if $t$ is constant 
and $1/poly(n) <  \delta < 1$. For unrestricted values of $t$, this problem 
is known to be complete for the complexity class $\mathrm{QCMA}$, a quantum generalization of MA.  
In contrast, we show that the same problem is $\mathrm{NP}$-complete if $t=O(\log n)$ even when $\delta$ is constant. 

On the other hand, we show that given a $n$-input quantum circuit $C$ of treewidth $t=O(\log n)$, and a constant $\delta<1/2$, 
it is $\mathrm{QMA}$-complete to determine whether there exists a quantum state $\ket{\varphi}\in (\C^d)^{\otimes n}$
such that the acceptance probability of $C\ket{\varphi}$ is greater than $1-\delta$, or whether for every 
such state $\ket{\varphi}$, the acceptance probability of $C\ket{\varphi}$ is less than $\delta$. 
As a consequence, under the widely believed assumption that $\mathrm{QMA} \neq \mathrm{NP}$, we have that
quantum witnesses are strictly more powerful than classical witnesses with respect to 
Merlin-Arthur protocols in which the verifier is a quantum 
circuit of logarithmic treewidth. 
\end{abstract}

\section{Introduction}
\label{section:Introduction}

The notions of tree decomposition and treewidth of a graph \cite{RobertsonSeymour1984}  
play a central role in algorithmic theory. On the one hand, many natural classes 
of graphs have small treewidth. For instance, trees have treewidth at most $1$, series-parallel graphs  and 
outer-planar graphs have treewidth at most $2$, Halin graphs have treewidth at most $3$, and 
$k$-outerplanar graphs for fixed $k$ have treewidth $O(k)$. 
On the other hand, many problems that are hard for $\mathrm{NP}$ on general graphs, and even problems that 
are hard for higher levels of the polynomial hierarchy, may be solved in 
polynomial time when restricted to graphs of constant tree-width 
\cite{ArnborgLagergrenSeese1991,ArnborgProskurowski1989,Courcelle1990Monadic}.
In particular, during the last decade, several algorithms 
running in time $2^{O(t)}\cdot n^{O(1)}$ have been proposed for the satisfiability of 
classical circuits\footnote{In the case of classical circuits, it is assumed that each variable labels a unique input  of 
unbounded fan-out.} and boolean constraint satisfaction problems of size $n$ and treewidth $t$ 
\cite{AlekhnovichRazborov2002,AllenderChenLouPeriklisPapakonstantinouTang2014,BroeringLokamSatyanarayana2004,GeorgiouKonstantinosPapakonstantinou2008}.

In this work, we identify for the first time a natural quantum optimization problem that becomes feasible 
when restricted to graphs of constant treewidth. More precisely, we show how to find in polynomial time 
a classical assignment that maximizes, up to an inverse polynomial additive factor, the acceptance probability 
of a quantum circuit of constant treewidth. 
For quantum circuits of unrestricted treewidth this problem is complete for $\mathrm{QCMA}$, a quantum generalization of 
MA \cite{AharonovNaveh2002}. 
Before stating our main result, we fix some notation. If $C$ is a quantum circuit acting on $n$ $d$-dimensional qudits, 
and $\ket{\psi}$ is a quantum state in $(\C^{d})^{\otimes n}$, then we denote by $\mathit{Pr}(C,\ket{\psi})$ the probability 
that the state of the output of $C$ collapses to $\ket{1}$ when the input of $C$ is initialized with 
$\ket{\psi}$ and the output is measured in the standard basis $\{\ket{0},\ket{1},...,\ket{d-1}\}$. If $y$ is a string in 
$\{0,...,d-1\}^{n}$ then we let $\ket{y}=\otimes_{i=1}^{n}\ket{y_i}$ denote the basis state corresponding to 
$y$. We let $\mathit{Pr}^{cl}(C) = \max_{{y\in \{0,...,d-1\}^n}} \mathit{Pr}(C, \ket{y})$ 
denote the maximum acceptance probability of $C$ among all classical input strings in $\{0,...,d-1\}^n$.
The treewidth of a quantum circuit is defined as the treewidth of its underlying undirected graph. 

\begin{theorem}[Main Theorem]
\label{theorem:CircuitSatisfiability}
Let $C$ be a quantum circuit with $n$ uninitialized inputs, $\mathit{poly}(n)$ gates, and treewidth $t$. 
For each $\delta$ with  $1/\mathit{poly}(n)<\delta< 1$, one may find in time $(\frac{n}{\delta})^{\exp(O(t))}$ 
a string $y\in \{0,...,d-1\}^{n}$ 
such that ${|\mathit{Pr}(C,\ket{y}) - \mathit{Pr}^{cl}(C)|\leq \delta}$.
\end{theorem}

We note that the algorithm that finds the string $y\in \{0,1\}^n$ in Theorem \ref{theorem:CircuitSatisfiability} 
is completely deterministic.
The use of treewidth in quantum algorithmics was pioneered by Markov and Shi \cite{MarkovShi2008} who showed that 
quantum circuits of logarithmic treewidth 
can be simulated in polynomial time with exponentially high precision. Note that the simulation of quantum circuits 
\cite{Gottesman1998,JozsaLinden2003,MarkovShi2008,Valiant2002,Vidal2003} 
deals with the problem of computing the acceptance probability of a quantum circuit when all inputs 
are already initialized, and thus may be regarded as a generalization of the classical P-complete problem CIRCUIT-VALUE. 
On the other hand, Theorem \ref{theorem:CircuitSatisfiability} deals with the problem of finding a classical assignment 
that maximizes the acceptance probability of a quantum circuit with uninitialized inputs, and thus may be regarded as a generalization 
of the classical $\mathrm{NP}$-complete problem CIRCUIT-SAT. In this sense, Theorem \ref{theorem:CircuitSatisfiability} is 
the first result showing that a quantum generalization of CIRCUIT-SAT can be solved in polynomial time when restricted to circuits of 
constant treewidth.

It is interesting to determine whether the time complexity of our algorithm can be substantially improved. To 
address this question, we first introduce the {\em online-width} of a circuit, a width measure for DAGs that is at least 
as large as the treewidth of their underlying undirected graphs.
If $G=(V,E)$ is a directed graph and $V_1,V_2\subseteq V$ are two subsets of vertices of $V$ with $V_1\cap V_2 = \emptyset$ then 
we let $E(V_1,V_2)$ be the set of all edges with one endpoint in $V_1$ and another endpoint in $V_2$. If $\omega=(v_1,v_2,...,v_n)$ 
is a total ordering of the vertices in $V$, then we let $\cutwidth(G,\omega) = \max_{i}|E(\{v_1,...,v_i\},\{v_{i+1},...,v_n\})|$. 
The {\em cutwidth} of $G$  is defined as $\cutwidth(G) = \min_{\omega}\cutwidth(G,\omega)$ where the 
minimum is taken over all possible total orderings of the vertices of $G$ \cite{BodlaenderThilikosSerna2000}.  
If $G$ is a DAG, then the {\em online-width} of $G$ is defined as  $\onlinecutwidth(G) = \min_{\omega}\cutwidth(G,\omega)$ where 
the minimum is taken only among the {\em topological orderings} of $G$. Treewidth, cutwidth and online-width 
are compared as follows\footnote{A proof that $\treewidth(G)\leq \cutwidth(G)$ can be found in \cite{Bodlaender1988}.}. 

\vspace{-10pt}
\begin{equation}
\label{equation:ComparisonMeasures}
\treewidth(G)  \leq \cutwidth(G) \leq \onlinecutwidth(G) 
\end{equation}

Theorem \ref{theorem:LogarithmicTreewidthNP} below states that finding a classical assignment that maximizes 
the acceptance probability of a quantum circuit of logarithmic online-width is $\mathrm{NP}$-complete even when $\delta$ is constant. 
We note that the same completeness result holds with respect to circuits of logarithmic treewidth.

\begin{theorem}
\label{theorem:LogarithmicTreewidthNP}
For any constant $\delta$ with $0<\delta < 1$, the following problem is $\mathrm{NP}$-complete:
Given a quantum circuit $C$ of online-width $O(\log n)$ with $n$ uninitialized inputs and $\mathit{poly}(n)$ gates, 
determine whether $Pr^{cl}(C)=1$ or whether $Pr^{cl}(C)\leq \delta$. 
\end{theorem}

An analog completeness result holds when the verifier is restricted to have logarithmic online-width and 
the witness is allowed to be an arbitrary quantum state. It was shown by Kitaev \cite{KitaevShenVyalyi2002}
that finding a $\delta$-optimal quantum witness for a quantum circuit of unrestricted width
is complete for the complexity class $\mathrm{QMA}$ for any constant $\delta$. Interestingly, Kitaev's completeness result 
is preserved when the quantum circuits are restricted to have logarithmic online-width. If $C$ is a quantum circuit 
with $n$ inputs, then 
we let $\Pr^{qu}(C)=\max_{\ket{\varphi}} Pr(C,\ket{\psi})$ be the maximum acceptance probability among 
all $n$-qudit quantum states $\ket{\psi}$.  

\begin{theorem}
\label{theorem:LogarithmicTreewidthQMA}
For any constant $\delta$ with $0<\delta< 1/2$, the following problem is $\mathrm{QMA}$-Complete:
Given a quantum circuit $C$ of online-width $O(\log n)$ with $n$ uninitialized inputs 
and $\poly(n)$ gates, determine whether ${Pr^{qu}(C)\geq 1-\delta}$ or whether ${Pr^{qu}(C)\leq \delta}$.
\end{theorem}

We analyse the implications of theorems \ref{theorem:LogarithmicTreewidthNP} and \ref{theorem:LogarithmicTreewidthQMA} to 
quantum generalizations of Merlin-Arthur protocols. 
In this setting, Arthur, a polynomial 
sized quantum circuit, must decide the membership of a string 
$x$ to a given language $\lang$ by analysing a quantum 
state $\ket{\psi}$ provided by Merlin. In the case that $x\in \lang$,
there is always a quantum state $\ket{\psi}$ that is accepted by 
Arthur with probability at least $2/3$. On the other hand, if $x\notin \lang$
then no state is accepted by Arthur with probability greater than 
$1/3$. The class of all languages 
that can be decided via some quantum Merlin-Arthur protocol is denoted 
by $\mathrm{QMA}$. The importance of $\mathrm{QMA}$ stems from the fact that this class has 
several natural complete problems \cite{Bookatz2014,KitaevShenVyalyi2002}. Additionally, 
the oracle version of $\mathrm{QMA}$ contains problems, such as the group non-membership problem \cite{Watrous2000} which are provably 
not in in the oracle version of MA and hence not in the oracle version of $\mathrm{NP}$ \cite{Babai1992}. The class $\mathrm{QCMA}$ is defined 
analogously, except for the fact that the witness provided by Merlin 
is a product state encoding a classical string. Below we define width parameterized versions of $\mathrm{QMA}$.

\begin{definition}
\label{definition:WidthQMA}
A language $\lang \subseteq \{0,1\}^*$ belongs to the class $\mathrm{QMA}[\treewidth,f(n)]$ 
if there exists a polynomial time constructible family of quantum circuits $\{C_x\}_{{x\in \{0,1\}^*}}$ such that 
 for every $x\in \{0,1\}^*$, $C_x$ has treewidth at most $f(|x|)$ and 
\begin{itemize}
\item if $x\in \lang$ then there exists a quantum state $\ket{\psi}$ such that $C_x$ accepts $\ket{\psi}$
		with probability at least $2/3$,
\item if $x\notin \lang$ then for each quantum state $\ket{\psi}$, $C_x$ accepts $\ket{\psi}$ with 
	probability at most $1/3$. 
\end{itemize}
The class $\mathrm{QCMA}[\treewidth,f(n)]$ is defined analogously, except that the witness $\ket{y}$ is
required to be the basis state encoding of a classical string $y$. 
\end{definition}

Definition \ref{definition:WidthQMA} can be extended naturally to other width measures such as online-width. 
For instance, $\mathrm{QMA}[\onlinecutwidth,f(n)]$ and $\mathrm{QCMA}[\onlinecutwidth,f(n)]$ denote 
the classes of languages that can be decided by quantum Merlin-Arthur games with respectively quantum and classical 
witnesses, in which the verifier is required to have online-width at most $f(n)$.
We note that the classes $\mathrm{QMA}$ and $\mathrm{QCMA}$ 
can be defined respectively as $\mathrm{QMA}[\onlinecutwidth,\poly(n)]$ and $\mathrm{QCMA}[\onlinecutwidth,\poly(n)]$,
since the online-width of a circuit can be at most quadratic in its number of gates. In the next corollary 
we analyse the complexity of low-width quantum Merlin-Arthur protocols with classical and quantum witnesses.

\begin{corollary}
\label{corollary:LogarithmicTreewidthQMAQCMA}
$ $
\begin{enumerate}[i.] 
	\item \label{corollary:O1} $\mathrm{QCMA}[\treewidth,O(1)]$ $\subseteq$ P.
	\item \label{corollary:NP} $\mathrm{QCMA}[\treewidth,O(\log n)] = \mathrm{QCMA}[\onlinecutwidth,O(\log n)] = \mathrm{NP}$.
	\item \label{corollary:QMA} $\mathrm{QMA}[\treewidth,O(\log n)] = \mathrm{QMA}[\onlinecutwidth,O(\log n)] = \mathrm{QMA}$. 
\end{enumerate}
\end{corollary}

We note that Corollary \ref{corollary:LogarithmicTreewidthQMAQCMA}.\ref{corollary:O1} is a consequence of 
Theorem \ref{theorem:CircuitSatisfiability}, Corollary \ref{corollary:LogarithmicTreewidthQMAQCMA}.\ref{corollary:NP} 
is a consequence of Theorem \ref{theorem:LogarithmicTreewidthNP}, and Corollary 
\ref{corollary:LogarithmicTreewidthQMAQCMA}.\ref{corollary:QMA} is a consequence of Theorem \ref{theorem:LogarithmicTreewidthQMA}.
Under the plausible assumption that $\mathrm{QMA}\neq \mathrm{NP}$, Corollary \ref{corollary:LogarithmicTreewidthQMAQCMA} implies that 
whenever Arthur is restricted to be a quantum circuit of logarithmic treewidth, 
quantum Merlin-Arthur protocols differ in power with respect to whether the witness provided by Merlin is classical or quantum. 
We observe that obtaining a similar separation between the power of classical and quantum witnesses 
when Arthur is allowed to be a quantum circuit of polynomial treewidth is equivalent to determining whether 
$\mathrm{QMA}\neq \mathrm{QCMA}$. 
This question remains widely open.

\subsection{Organization of the Paper}
\label{subsection:ProofOverview}

In Section \ref{section:Preliminaries} we will define basic notions such as quantum circuits, tree decompositions 
and treewidth. Sections \ref{section:AbstractNetwork} to \ref{section:ApproximatingFeasibility} will be dedicated 
to the proof of our main theorem (Theorem \ref{theorem:CircuitSatisfiability}). The proof of
this theorem will be sketched in Subsection \ref{subsection:ProofSketch}. 
In Section \ref{section:ProofLogarithmicTreewidthNP} we will prove Theorem \ref{theorem:LogarithmicTreewidthNP}, and 
in Section \ref{section:ProofLogarithmicTreewidthQMA} we will prove Theorem \ref{theorem:LogarithmicTreewidthQMA}. 
We will conclude this paper by making some final considerations and by stating some open problems in 
Section \ref{section:Conclusion}. 

\subsection{Sketch of the Proof of Theorem \ref{theorem:CircuitSatisfiability}}
\label{subsection:ProofSketch}

We will prove Theorem \ref{theorem:CircuitSatisfiability}  using a combination of techniques 
from tensor network theory, structural graph theory and dynamic programming. 
We will start by introducing in Section \ref{section:AbstractNetwork} the notion of 
{\em abstract network}. Intuitively, an abstract network is a list $\abstractnetwork=\{\indexset_1,...,\indexset_n\}$ of 
finite subsets of positive integers, called {\em index sets}. Such an abstract network $\abstractnetwork$ 
can be naturally associated with a graph $\graph(\abstractnetwork)$. This graph is obtained by creating a vertex 
$v_{\indexset}$ for each index set $\indexset\in \abstractnetwork$ and by adding $k$ edges between two vertices 
$v_{\indexset}$ and $v_{\indexset'}$ if and only if 
$|\indexset\cap \indexset'|=k$. We will be interested in the process of contracting the vertices 
of the graph $\graph(\abstractnetwork)$ into a single vertex. Such a contraction process will be represented 
by a data structure called {\em contraction tree}. The complexity of a contraction tree will be measured 
via two parameters: its {\em rank}, and its {\em height}. 
In Section \ref{section:GoodContractionTree} we will show that if the graph $\graph(\abstractnetwork)$ 
associated with an abstract network $\abstractnetwork$ has treewidth $t$ and maximum degree $\Delta$, then 
one can efficiently construct a contraction tree for $\abstractnetwork$ of rank $O(\Delta\cdot t)$ and 
height $O(\Delta\cdot t\cdot \log|\abstractnetwork|)$ (Theorem \ref{theorem:GoodContractionTree}). 
As we will argue below, Theorem \ref{theorem:GoodContractionTree} will play an important role 
in the proof of our main result. 

Abstract networks can be used to define both the well known notion of {\em tensor network} (Section \ref{section:TensorNetwork}),
and the new notion of {\em feasibility tensor network} (Section \ref{section:FeasibilityTensorNetworks}). Within 
this formalism, a tensor network can be viewed as a pair $(\abstractnetwork,\lambda)$ where $\lambda$ is a 
function that associates a tensor 
$\lambda(\indexset)$ of rank $|\indexset|$ with each index set $\indexset$ of $\abstractnetwork$. 
One can define a notion of contraction for tensor networks with basis on the notion of contraction for 
abstract networks. Contracting a tensor network $(\abstractnetwork,\lambda)$ yields a complex number $g$, i.e., a 
tensor of rank $0$. The 
value of the tensor network, denoted by $\valuetensornetwork(\abstractnetwork,\lambda)$, is defined as 
the absolute value of $g$. It can be shown that the problem of computing the acceptance 
probability of a quantum circuit $C$ in which all inputs are initialized can be reduced to the problem of computing 
the value of a suitable tensor network $(\abstractnetwork_C,\lambda_C)$. 

On the other hand, a feasibility tensor network is a 
pair $(\abstractnetwork,\Lambda)$ where $\abstractnetwork$ is an abstract network and $\Lambda$ is a
function that associates with each index set $\indexset\in \abstractnetwork$ a {\em set} of tensors $\Lambda(\indexset)$
of rank $|\indexset|$. An initialization of $(\abstractnetwork,\Lambda)$ is a function $\lambda$ that associates
a tensor $\lambda(\indexset)\in \Lambda(\indexset)$ with each index set of $\abstractnetwork$.
Each such initialization yields a tensor network $(\abstractnetwork,\lambda)$. 
The value of the feasibility network $(\abstractnetwork,\Lambda)$, is defined as 
$\valuefeasibilitytensornetwork(\abstractnetwork,\Lambda) = \max_{\lambda} \valuetensornetwork(\abstractnetwork,\lambda)$, 
where $\lambda$ ranges over all initializations of $(\abstractnetwork,\Lambda)$. As we will see in Section 
\ref{section:FeasibilityTensorNetworks}, the problem of finding a classical assignment that maximizes 
the acceptance probability of a quantum circuit with uninitialized inputs can be reduced to the problem
of finding an initialization of maximum value for a suitable feasibility tensor network $(\abstractnetwork_C,\Lambda_C)$. 

Let $(\abstractnetwork,\Lambda)$ be a feasibility tensor network, and $\varepsilon$ be a real number
with ${0<\varepsilon < 1}$. In Section \ref{section:ApproximatingFeasibility} we will show that 
given a contraction tree for $\abstractnetwork$ of rank $r$ and height $h$, one can 
find in time $|\abstractnetwork| \cdot \varepsilon^{-\exp(O(r\cdot \log d))}$ an 
initialization $\lambda$ of $(\abstractnetwork,\Lambda)$ such that 
$|\valuetensornetwork(\abstractnetwork,\lambda) - \valuefeasibilitytensornetwork(\abstractnetwork,\Lambda)|\leq 
\varepsilon\cdot \exp(O(r\cdot h\cdot \log d))$ (Theorem \ref{theorem:TensorNetworkSatisfiability}). 
Therefore, to obtain a polynomial time algorithm for approximating the value of a 
feasibility tensor network $(\abstractnetwork,\Lambda)$ up to a constant additive factor $\delta$, we need to 
keep the rank of the contraction tree bounded by a constant, and its height bounded by $O(\log |\abstractnetwork|)$.

As mentioned above, if the graph $\graph(\abstractnetwork)$ has treewidth $t$ and maximum degree $\Delta$, 
then by Theorem \ref{theorem:GoodContractionTree} one can construct a contraction tree for $\abstractnetwork$ of 
rank $O(\Delta\cdot t)$, and height $O(\Delta\cdot t\cdot \log |\abstractnetwork|)$. Therefore, by setting 
$\varepsilon = \delta/|\abstractnetwork|^{O(\Delta^2\cdot t^2\cdot \log d)}$ we can use 
Theorem \ref{theorem:TensorNetworkSatisfiability} to 
find in time $(|\abstractnetwork|/\delta)^{\exp(O(\Delta\cdot t\cdot \log d))}$ an initialization 
$\lambda$ for $(\abstractnetwork,\Lambda)$ such that 
$|\valuetensornetwork(\abstractnetwork,\lambda)- \valuetensornetwork(\abstractnetwork,\Lambda)|\leq \delta$ 
(Theorem \ref{theorem:ApproximationFeasibility}).

Finally, let $C$ be a quantum circuit with $n$ uninitialized inputs, treewidth $t$, and $n^{O(1)}$ gates drawn from 
a finite universal set of gates $\universalgates$. Let $(\abstractnetwork_C,\Lambda_C)$ be the feasibility 
tensor network associated with $C$. Then $|\abstractnetwork|=n^{O(1)}$ and the graph $\graph(\abstractnetwork_C)$ 
has treewidth $t$, and maximum degree bounded by a constant $\Delta(\universalgates)$. 
Therefore, as a corollary of Theorem \ref{theorem:ApproximationFeasibility}, 
a classical assignment that maximizes the acceptance probability of $C$ up to a $\delta$ additive factor can 
be found in time $(n/\delta)^{\exp(O(\Delta(\universalgates)\cdot t\cdot \log d))}$. Since both $\Delta(\universalgates)$
and the dimensionality $d$ of the qudits over which $C$ operates are constant, Theorem \ref{theorem:CircuitSatisfiability} follows.

\section{Preliminaries}
\label{section:Preliminaries}

A $d$-dimensional qudit is a unit vector in the Hilbert space $\hilbert_d= \C^d$. We fix an orthonormal basis 
for $\hilbert_d$ and label the vectors in this basis with ${\ket{0},\ket{1},...,\ket{d-1}}$.
The $n$-fold tensor product of $\hilbert_d$ is denoted by $\hilbert_d^{\otimes n}$.
We denote by $\operators(\hilbert_d^{\otimes n})$ the set of all linear operators 
on $\hilbert_d^{\otimes n}$. 
An operator $X$ on $\operators(\hilbert_d^{\otimes n})$ is positive semidefinite if all its eigenvalues are non-negative. 
A density operator on $n$ qudits is a positive semidefinite operator 
$\rho\in \mathbf{L}(\hilbert_d^{\otimes n})$ with trace $\trace(\rho)=1$. 
For a string $y=y_1y_2...y_n\in \{0,1,...,d-1\}^n$ 
we let $\rho_{y} = \bigotimes_{i=1}^n \ket{y_i}\bra{y_i}$ be the density operator of the state 
${\ket{y} = \otimes_{i=1}^n\ket{y_i}}$. 
A map $M:\operators(\hilbert_d^{\otimes q})\rightarrow \operators(\hilbert_d^{\otimes r})$ is positive if 
$M(\rho)$ is positive semidefinite whenever $\rho$ is positive semidefinite. The map $M$ 
is completely positive if the map $I_k\otimes M$ is positive for every $k\in \N$, where $I_k$ is 
the $k\times k$ identity matrix. 
A quantum gate with $q$ inputs and $r$ outputs is a  linear map 
$Q:\mathbf{L}(\hilbert_d^{\otimes q}) \rightarrow \mathbf{L}(\hilbert_d^{\otimes r})$ that is 
completely positive, convex on density matrices, and such that $0\leq \trace(Q(\rho))\leq 1$ for any density 
matrix $\rho$. Linear maps satisfying these three properties formalize the notion of physically 
admissible quantum operation. We refer to \cite{NielsenChuang2010} (Section 8.2.4) for 
a detailed discussion on physically admissible operations. 
A {\em positive-operator valued measure} (POVM) is a set $\mathcal{X}=\{X_1,X_2,...,X_k\}$ of positive semidefinite 
operators such that $\sum_i X_i = I$. Each operator $X_i$ in $\mathcal{X}$ is called a {\em measurement element} of $\mathcal{X}$. 
If $\mathcal{X}$ is a POVM then the probability of measuring outcome $i$ after applying $\mathcal{X}$ to $\rho$ is given 
by $tr(\rho X_i)$. A single $d$-dimensional qudit measurement in the computational basis is defined as 
the POVM $\mathcal{X}=\{\ket{0}\bra{0},\ket{1}\bra{1},...,\ket{d-1}\bra{d-1}\}$. 

\subsection{Quantum Circuits}
\label{subsection:QuantumCircuits}

We adopt the model of quantum circuits with mixed states introduced in \cite{AharonovKitaevNisan1998}. 
Let $\universalgates$ be a finite universal set of quantum gates, and let $\Delta(\universalgates)$ be the maximum 
number of inputs plus outputs of a gate in $\universalgates$.  
A quantum circuit over $\universalgates$ is a connected directed acyclic 
graph\footnote{All graphs in this work, being directed or undirected, may contain multiple edges, but no loops.} 
$C = (V,E,\vertexlabelingfunction,\edgelabelingfunction)$,
of maximum degree at most $\Delta(\universalgates)$, where $V$ is a set of vertices, $E$ a set of edges, 
${\vertexlabelingfunction:V\rightarrow \universalgates}$ 
is a vertex labeling function and 
${\edgelabelingfunction:E\rightarrow \{1,...,|E|\}}$ is an injective function that assigns a distinct number to each 
edge of $C$. The vertex set is partitioned into a set $\inputvertices$ (input vertices), a set $\outputvertices$ 
(output vertices), and a set ${\internalvertices = V\backslash (\inputvertices \cup \outputvertices)}$ (internal vertices).
Each input vertex has in-degree $0$ and out-degree $1$, and each output vertex has in-degree $1$ and out-degree $0$. Each internal vertex 
has both in-degree and out-degree greater than $0$. If $v$ is an internal vertex with $k$ incoming edges and $l$ outgoing edges 
then $v$ is labeled with a quantum gate $\vertexlabelingfunction(v)\in \universalgates$ with $k$ inputs and $l$ outputs. Each input vertex $v$ is either 
labeled by $\vertexlabelingfunction$ with an element from the set ${\{\ket{0}\bra{0},\ket{1}\bra{1},...,\ket{d-1}\bra{d-1}\}}$,
indicating that $v$ is an initialized 
input, or with the symbol $*$, indicating that $v$ is not initialized. Finally, each output vertex $v$ is labeled 
with an one-qudit measurement element $\vertexlabelingfunction(v)\in \operators(\hilbert_d)$. We let $M(C)=\otimes_{v\in \outputvertices} \vertexlabelingfunction(v)$ 
denote the overall measurement element in $\operators(\hilbert_d^ {\otimes |\outputvertices(C)|})$ defined by $C$.
A quantum circuit $C$ with $n$ uninitialized inputs and $m$ outputs can be regarded as a superoperator 
$C:\operators(\hilbert_d^ {\otimes n}) \rightarrow \operators(\hilbert_d^{\otimes m})$. 
If $\ket{\psi}$ is a quantum state in $\hilbert_d^{\otimes n}$ then the acceptance probability of $C$ when 
$\ket{\psi}$ is assigned to the inputs of $C$ is defined as $Pr(C,\ket{\psi})  = \trace[C(\ket{\psi}\bra{\psi})\cdot M(C)]$.

\subsection{Tree Decompositions and Treewidth}
\label{subsection:TreeDecomposition}

A tree is a connected acyclic graph $T$ with set of nodes $\nodes(T)$ and set of arcs $\arcs(T)$. 
A tree decomposition of a graph $G = (V,E)$ consists of a pair $(T,\beta)$
where $T$ is a tree, and $\beta:\nodes(T)\rightarrow 2^{V}$ is a function that associates a set 
of vertices $\beta(u)$ with each node $u\in \nodes(T)$, in such a way that  

\begin{itemize}
	\item $\bigcup_{u\in \nodes(T)} \beta(u) =  V$, 
	\item for every edge $\{v,v'\} \in E$, there is a node $u\in \nodes(T)$ such that ${\{v,v'\} \subseteq \beta(u)}$,
	\item for every vertex $v\in V$, the set $\{u\in \nodes(T) \;|\; v\in \beta(u)\}$ induces a connected 
		subtree of $T$. 
\end{itemize}
The {\em width} of $(T,\beta)$ is defined as  $\width(T,\beta)=\max_{u}\{|\beta(u)|-1\}$.
The {\em treewidth} $\treewidth(G)$ of a graph $G$ is the minimum width of a tree decomposition of $G$.

If $C = (V,E,\vertexlabelingfunction,\edgelabelingfunction)$ is a quantum circuit, then the 
treewidth of $C$ is defined as the treewidth of the undirected graph $G_C = (V,E')$ obtained from 
$C$ by forgetting vertex labels, edge labels, and direction of edges.

\section{Abstract Networks}
\label{section:AbstractNetwork}

In this section we will introduce the notion of {\em abstract network}. In Section \ref{section:TensorNetwork} 
we will use abstract networks to model the well known notion of tensor network, a formalism that 
is suitable for the simulation of quantum circuits. Subsequently, in Section \ref{section:FeasibilityTensorNetworks},
we will use abstract networks to define the new notion of {\em feasibility tensor network}, a formalism
that is suitable for addressing the satisfiability of quantum circuits. 
Below, we call a possibly empty finite set $\indexset$
of positive integers, an {\em index set}. We say that each number $i$ in an index set $\indexset$ is an {\em index}.

\begin{definition}[Abstract Network]
\label{definition:AbstractNetwork}
An {\em abstract network} is a finite list $${\abstractnetwork= [\,\indexset_1,...,\indexset_m\,]}$$ of 
index sets satisfying the following property:
\begin{equation}
\label{equation:ConditionNetwork}
 \forall i\in \bigcup_{k=1}^{m} \indexset_k, \;\;|\{j \;|\; i\in \indexset_j \}| = 2.
\end{equation} 
\end{definition}

In other words, in an abstract network $\abstractnetwork$, each index $i$ occurs in precisely two index sets of $\abstractnetwork$.
We note that an index set $\indexset$ can occur up to two times in an abstract network. We let $|\abstractnetwork|$ denote 
the size of $\abstractnetwork$, i.e., $m$. The rank of $\abstractnetwork$, denoted by $\rank(\abstractnetwork)$, is defined 
as the size of the largest index set in $\abstractnetwork$. 

\begin{equation}
\label{equation:Rank}
\rank(\abstractnetwork) = \max_{j\in \{1,...,m\}} |\indexset_j|.
\end{equation}

An abstract network $\abstractnetwork$ can be intuitively visualized as a graph $\graph(\abstractnetwork)$ 
which has one vertex $v_{\indexset}$ for each index set $\indexset\in \abstractnetwork$, and one edge $e$ 
with endpoints $\{\indexset,\indexset'\}$ and label $i$, for each pair of index sets $\indexset,\indexset$ with $\indexset\cap \indexset' \neq \emptyset$
and each index $i\in \indexset\cap \indexset'$ (Fig. \ref{figure:AbstractNetworkContraction}). Note
that our notion of graph of an abstract network admits multiple edges, but no loops. 
 We say that an abstract network $\abstractnetwork$ is connected if the graph $\graph(\abstractnetwork)$ 
associated with $\abstractnetwork$ is connected. In this work we will only be concerned with connected abstract networks.   

\begin{figure}[htb]
\centering
\includegraphics[scale=0.23]{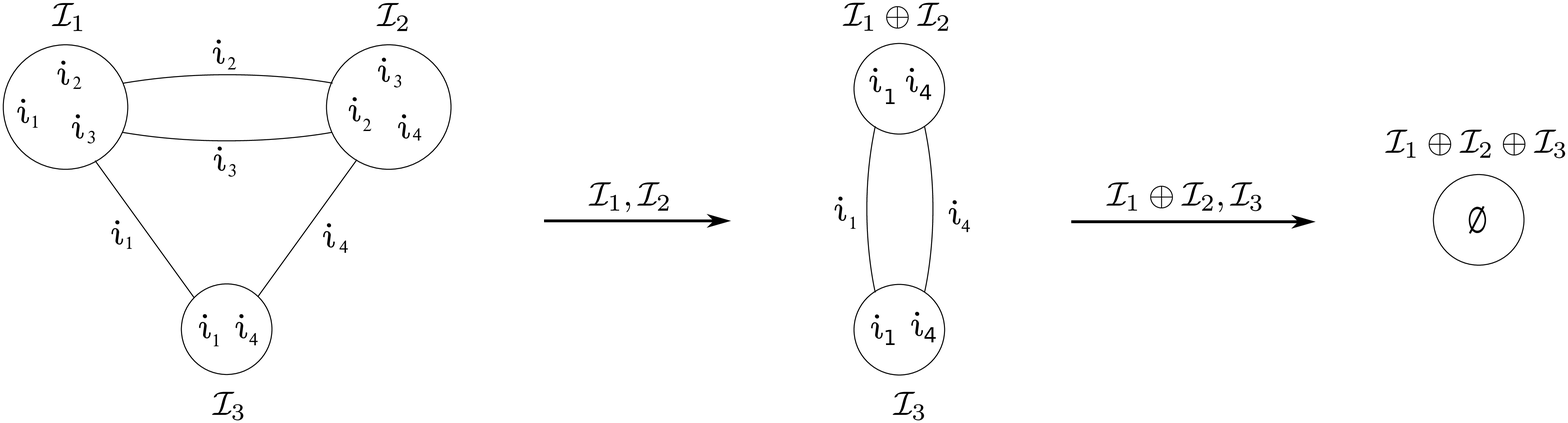}
\caption{Left: the graph $\graph(\abstractnetwork)$ of an abstract network 
$\abstractnetwork =[\,\indexset_1,\indexset_2,\indexset_3 \,]$. Middle:
contracting the index sets $\indexset_1$ and $\indexset_2$ yields the abstract network 
$\abstractnetwork =[\, \indexset_3,\indexset_1\oplus \indexset_2\,]$. 
Right: after all pairs have been contracted, the only remaining index set is the empty index set.
} 
\label{figure:AbstractNetworkContraction}
\end{figure}

There is a very simple notion of contraction for abstract networks. Abstract network contractions will be used to 
formalize both the well known notion of tensor network contraction (Section \ref{section:TensorNetwork}), and the 
notion of feasibility tensor network contraction, which will be introduced in Section \ref{section:FeasibilityTensorNetworks}. 
We say that a pair 
of index sets $\indexset,\indexset'$ of an abstract network $\abstractnetwork$ is {\em contractible} if $\indexset\cap \indexset' \neq \emptyset$.
In this case the contraction of $\indexset,\indexset'$ yields
the abstract network $$\abstractnetwork'= \abstractnetwork \backslash \{\indexset,\indexset'\} \cup \{\indexset \oplus \indexset'\}$$
where $\indexset \oplus \indexset'=\indexset \cup \indexset' \backslash(\indexset \cap \indexset')$ is the symmetric difference of $\indexset$ and $\indexset'$.
The contraction of a pair of index sets in an abstract network $\abstractnetwork$ may be visualized as an operation that merges the vertices 
$v_{\indexset}$ and $v_{\indexset'}$ in the graph $G(\abstractnetwork)$ associated with $\abstractnetwork$ (Fig. \ref{figure:AbstractNetworkContraction}). 
Observe that in a connected abstract network with at least two vertices, there is at least one pair of contractible index sets.
Additionally, when contracting a pair $\indexset,\indexset'$ of index sets, Equation \ref{equation:ConditionNetwork} ensures 
that the index set $\indexset\oplus \indexset'$ is not in $\abstractnetwork$. Thus we have that $|\abstractnetwork'|= |\abstractnetwork|-1$.
Starting with an abstract network $\abstractnetwork$ we can successively contract pairs of index-sets until we reach an abstract network 
whose unique index set is the empty set $\emptyset$. In graph-theoretic terms, starting from $\graph(\abstractnetwork)$ we can successively 
merge pairs of adjacent vertices until we reach the graph $\graph([\,\emptyset\,])$ with a single vertex $v_{\emptyset}$ 
(Fig. \ref{figure:AbstractNetworkContraction}).  
Below we define the notion of contraction tree, which will be used to address both the problem of simulating an initialized quantum circuit, 
and the problem of computing the maximum acceptance probability of an uninitialized quantum circuit. If $T$ is a tree, we 
denote by $\leaves(T)$ the set of leaves of $T$. We say that a node $u\in \nodes(T)\backslash \leaves(T)$ is 
an {\em internal node} of $T$. 

\begin{definition}[Contraction Tree]
\label{definition:CarvingDecompositionAbstract}
Let  $\abstractnetwork = [\,\indexset_1,...,\indexset_m\,]$ be an abstract network. 
A {\em contraction tree} for $\abstractnetwork$ is a pair $(T,\iota)$ where $T$ is a binary tree 
and  ${\iota:\nodes(T) \rightarrow 2^{\N}}$ is a function that associates with each node $u\in \nodes(T)$,
an index set $\iota(u)$ such that the following conditions are satisfied. 
\begin{enumerate}[(i)]
\item \label{definition:CarvingDecompositionAbstract:ConditionI} $\leaves(T) = \{u_1,...,u_m\}$ and for 
each $j\in \{1,...,m\}$, $\iota(u_j) = \indexset_j$. 
\item \label{definition:CarvingDecompositionAbstract:ConditionII} For each internal node $u$, $\iota(u.l)\cap \iota(u.r) \neq \emptyset$ and $\iota(u)=\iota(u.l)\oplus \iota(u.r)$. 
\end{enumerate} 
\end{definition}

\begin{figure}[t]
\centering
\includegraphics[scale=0.23]{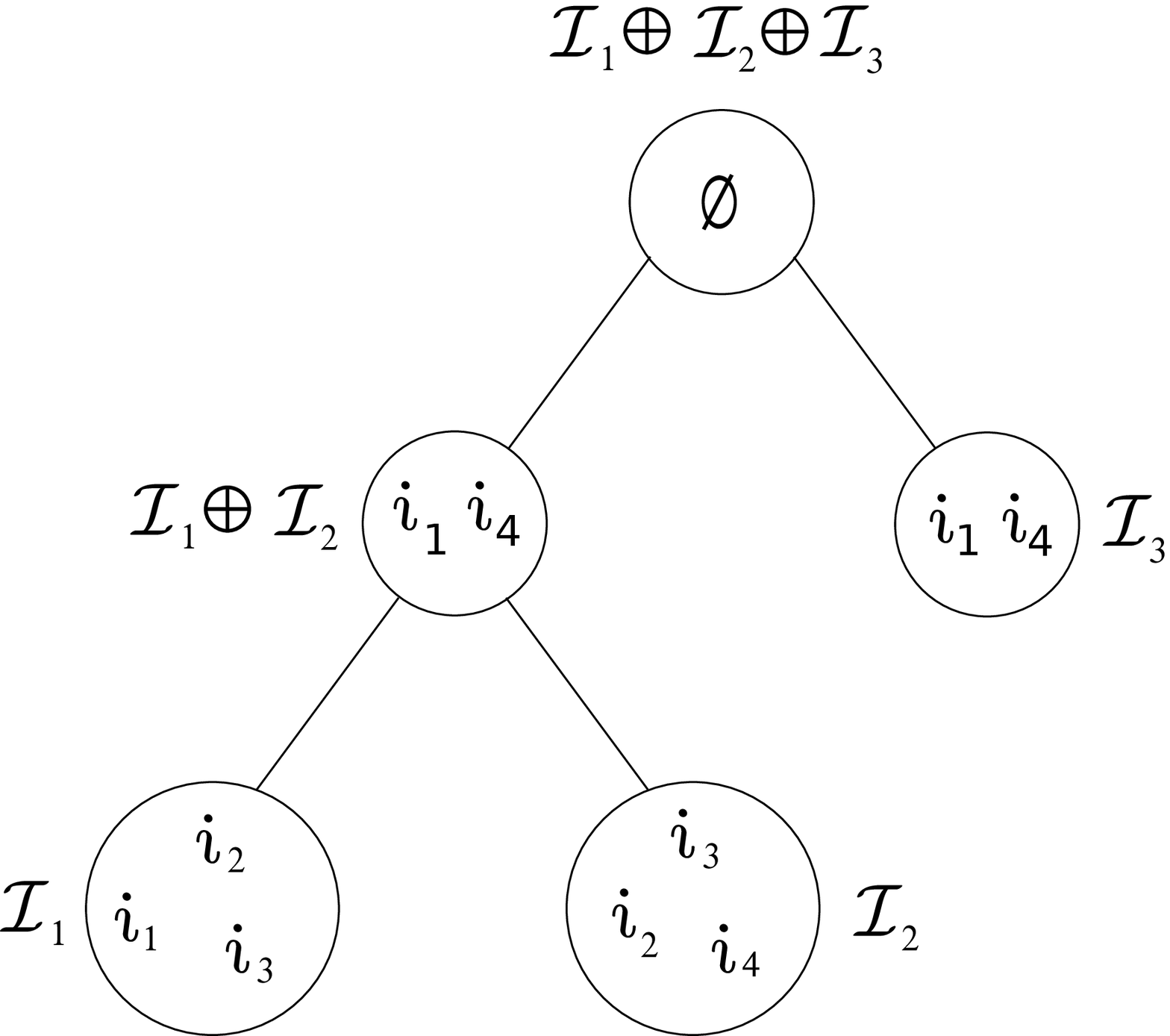}
\caption{A contraction tree of rank $3$ of the network $\abstractnetwork = \{\indexset_1,\indexset_2,\indexset_3\}$ of Fig. \ref{figure:AbstractNetworkContraction}.} 
\label{figure:ContractionTree}
\end{figure}

Intuitively, Condition $($\ref{definition:CarvingDecompositionAbstract:ConditionI}$)$ says that 
the restriction of $\iota$ to $\leaves(T)$ is a bijection from $\leaves(T)$ to the index sets occurring in 
$\abstractnetwork$, while Condition  $($\ref{definition:CarvingDecompositionAbstract:ConditionII}$)$ says that it is always 
possible to contract the index sets labeling the children of each internal node of $T$. 
Note that the root of $T$ is always labeled with the empty index set $\emptyset$. 
The rank of $(T,\iota)$ is the size of the largest index set labeling a node of $T$. $$\rank(T,\iota) = \max_{u\in \nodes(T)} |\iota(u)|.$$ 

\section{Contraction Trees of Constant Rank and Logarithmic Height}
\label{section:GoodContractionTree}

In this section we will show that if $\abstractnetwork$ is an abstract 
network whose graph $\graph(\abstractnetwork)$ has treewidth $t$ and 
maximum degree $\Delta$, then one can efficiently construct a 
contraction tree for $\abstractnetwork$ of rank $O(\Delta\cdot t)$ and 
height $O(\Delta\cdot t\cdot \log|\abstractnetwork|)$. More precisely, 
we will prove the following theorem. 

\begin{theorem}[Good Contraction Tree]
\label{theorem:GoodContractionTree}
Let $\abstractnetwork$ be an abstract network such that the graph $\graph(\abstractnetwork)$ has treewidth $t$ and 
maximum degree $\Delta$. Then one can construct in time $2^{O(t)}\cdot |\abstractnetwork|^{O(1)}$ a contraction tree $(T,\iota)$ for $\abstractnetwork$
of rank $O(\Delta \cdot t)$ and height $O(\Delta \cdot t \cdot \log |\abstractnetwork|)$. 
\end{theorem}

Below we define the notion of rooted carving decomposition, a variant 
of the notion of carving decomposition introduced by Robertson and Seymour in \cite{RobertsonSeymour1995}. 

\begin{definition}[Rooted Carving Decomposition]
\label{definition:CarvingDecomposition}
A {\em rooted carving decomposition} of a graph $G=(V,E)$ is a pair $(T,\gamma)$ where
$T=(N,F)$ is a rooted binary tree, and $\gamma:\leaves(T)\rightarrow V$ is a bijection mapping each leaf 
$u\in \leaves(T)$ to a single vertex $\gamma(u)\in V$. 
\end{definition}

Observe that the internal nodes of a carving decomposition $(T,\gamma)$ are unlabeled.
We denote by $T[u]$ the set of nodes of the subtree of $T$ rooted at $u$. 
Given a node $u\in \nodes(T)$ we let 
\begin{equation}
\label{equation:Vu}
V[u]=\gamma(\leaves(T[u]))=\{v\in V\;|\;\exists u\in \leaves(T[u]), \gamma(u)=v\}
\end{equation}
be the set of vertices of $G$ that are associated with some leaf in the subtree of $T$ rooted at $u$. For subsets of vertices $V_1,V_2\subseteq V$, let $E(V_1,V_2)$ 
denote the set of edges in $G$ with one endpoint in $V_1$ and another endpoint in $V_2$. 
The width  of $(T,\gamma)$, denoted $\carvingwidth(T,\gamma)$, is defined as 
\begin{equation}
\label{equation:CarvingWidth}
\carvingwidth(T,\gamma) = \max_{u\in \nodes(T)} |E(V[u],V\backslash V[u])|.
\end{equation}
The carving width of a graph $G$, denoted $\carvingwidth(G)$, is the minimum width of a carving decomposition of $G$.

Next, we establish some connections between tree-decompositions and carving decompositions. 
Let $G$ be a graph of treewidth $t$. Using the results in \cite{RobertsonSeymour1984,BodlaenderFominKosterKratschThilikos2012},
one can construct a tree decomposition $(T,\beta)$ of $G$ of width $O(t)$ in time $2^{O(t)}\cdot |G|^{O(1)}$.
From such a tree-decomposition $(T,\beta)$ one can construct in time $|T|^{O(1)}$ another tree-decomposition $(T',\beta')$ of $G$ of 
width $O(t)$ and height $O(\log |G|)$ \cite{Bodlaender1989}. Finally, from $(T',\beta')$ one can construct in time $|T'|^{O(1)}$ 
a carving decomposition of $G$ of width $O(\Delta\cdot t)$ and height $O(\log |G|)$ \cite{BodlaenderThilikosSerna2000} where 
$\Delta$ is the maximum degree of $G$. We formalize the series of conversions we have just described into the 
following lemma. 

\begin{lemma}[\cite{RobertsonSeymour1984,Bodlaender1989,BodlaenderThilikosSerna2000}]
\label{lemma:carvinwidthtreewidth}
Let $G$ be a graph of maximum degree $\Delta$ and treewidth $t$. One can construct in time
$2^{O(t)}\cdot|G|^{O(1)}$ a rooted carving decomposition of $G$ of width $O(\Delta\cdot t)$, and height $O(\log |G|)$. 
\end{lemma}

For the purposes of this work we need a more well behaved notion of carving decomposition, 
which we call {\em contractive carving decomposition}. 

\begin{definition}[Contractive Carving Decomposition]
\label{definition:ContractiveCarvingDecomposition}
We say that a rooted carving decomposition $(T,\gamma)$ of a graph $G=(V,E)$ is {\em contractive} if for each internal 
node $u$ of $T$,  $$E(V[u.l], V[u.r])\neq \emptyset.$$
\end{definition}

In other words, a rooted carving decomposition is {\em contractive} if for each internal node $u$ of $T$ there 
is at least one edge $e$ of $G$ such that one endpoint of $e$ labels a leaf of $T[u.l]$ and the other endpoint 
of $e$ labels a leaf of $T[u.r]$. The next lemma 
states that any rooted carving decomposition of width $w$ and height $h$ can be transformed into a {\em contractive} carving decomposition of 
width $w$ and height $w\cdot h$.
 
\begin{lemma}
\label{lemma:GoodCarving}
Let $G$ be a connected graph and $(T,\gamma)$ be a carving decomposition of $G$ of width $w$ and height $h$. Then one can construct in time 
$O(w\cdot |T|)$ a contractive carving decomposition of $G$ of width $w$ and height $w\cdot h$. 
\end{lemma}

We will prove Lemma \ref{lemma:GoodCarving} in Subsection \ref{subsection:ProofOfLemmaGoodCarving}. Before that, 
we will use this lemma to prove Theorem \ref{theorem:GoodContractionTree}.

\paragraph{\bf Proof of Theorem \ref{theorem:GoodContractionTree}.}
Let $\abstractnetwork$ be an abstract network such that ${\graph(\abstractnetwork)=(V,E)}$ 
has treewidth $t$ and maximum degree $\Delta$. Note that ${|\graph(\abstractnetwork)| = |\abstractnetwork|}$. 
By Lemma \ref{lemma:carvinwidthtreewidth}, we can construct a rooted carving decomposition of 
$\graph(\abstractnetwork)$ of width $O(\Delta\cdot t)$ and height $O(\log |\abstractnetwork|)$.
By Lemma \ref{lemma:GoodCarving} we can convert $(T,\gamma)$ into a contractive carving decomposition $(T',\gamma')$
of $\graph(\abstractnetwork)$ of width $O(\Delta\cdot t)$ and height $O(\Delta\cdot t\cdot \log|\abstractnetwork|)$. 

Now, we define a function $\iota:\nodes(T')\rightarrow 2^{\N}$ as follows. 
For each leaf $u$ of $T'$ labeled 
with the vertex $\gamma'(u) = v_{\indexset}\in \graph(\abstractnetwork)$, we 
set $\iota(u) = \indexset$. Therefore, at this point we have that $\iota$ establishes a bijection between leaves 
of $T'$ and index sets of $\abstractnetwork$. Next, for each internal node $u\in \nodes(T')$ we set 
$\iota(u)=\iota(u.l)\oplus\iota(u.r)$. Since $(T',\gamma')$ is a contractive carving decomposition of 
height $O(\Delta\cdot t\cdot \log |\abstractnetwork|)$, the pair $(T',\iota)$ is a contraction tree for $\abstractnetwork$ 
of height $O(\Delta\cdot t\cdot \log |\abstractnetwork|)$. 

We claim that the rank of $(T',\iota)$ is at most 
$O(\Delta\cdot t)$. To see this, note that for each node $u$ of $T'$,
$$\iota(u) =\bigoplus_{u'\in \leaves(T'[u])} \iota(u').$$

In other words, $\iota(u)$ is constituted by those indices that occur in precisely one leaf of the subtree $T'[u]$ rooted at $u$.

$$\iota(u) =  \left\{j\;|\;\mbox{There is a unique leaf $u'$ of $T'[u]$ such that $j\in \iota(u')$}\right\}.$$

But this implies that $\iota(u)$ is is precisely the set of indices labeling edges of $\graph(\abstractnetwork)$ 
which lie in $E(V[u],V\backslash V[u])$, where 
$$V[u] = \{v_{\indexset}\in \graph(\abstractnetwork)\;|\; \exists u'\in \leaves(T'[u]), \gamma(u') = v_{\indexset}\}.$$ 
Since each index $j\in \iota(u)$ labels a unique edge in 
$\graph(\abstractnetwork)$, we have that 
$${|\iota(u)| = |E(V[u],V\backslash V[u])|}.$$ 
Therefore, the rank of the contraction tree 
$(T',\iota)$ is equal to the width ${\carvingwidth(T',\gamma')}$ of 
the carving decomposition $(T',\gamma')$ (see Equation \ref{equation:CarvingWidth}). 
Since, by construction, ${\carvingwidth(T',\gamma') = O(\Delta\cdot t)}$, the claim follows. $\square$

\subsection{Proof of Lemma \ref{lemma:GoodCarving}} 
\label{subsection:ProofOfLemmaGoodCarving}

In this subsection we will prove Lemma \ref{lemma:GoodCarving}, which 
states that any rooted carving decomposition $(T,\gamma)$ of width $w$ and height $h$  of a graph $G$ 
can be transformed into a contractive carving decomposition of width $w$ and height $w\cdot h$.
Recall that if $(T,\gamma)$ is a rooted carving decomposition of a graph $G=(V,E)$ then for 
each node $u\in \nodes(T)$ we let $V[u]$  denote the set of nodes associated with the leaves in the subtree of 
$T$ rooted at $u$ (Equation \ref{equation:Vu}).
Below, we let $G[u]$ denote the subgraph of $G$ induced by the vertices in $V[u]$.  

\begin{proposition}
\label{proposition:DisjointSetsSubtree}
Let $G$ be a connected graph and $(T,\gamma)$ be a rooted carving decomposition of $G$ of width $w$.
For each node $u$ of $T$, the graph $G[u]$ has at most $w$ connected components.
\end{proposition}
\begin{proof}
Let $u$ be the root of $T$. Since $G$ is connected, $G[u]$ has a unique connected component, which 
is $G$ itself. Now let $u\in nodes(T)$ be a node which is not the root of $T$ and assume for contradiction that the connected components of 
$G[u]$ are $G_1,G_2,...,G_k$ for some $k>w$. Since $G$ is connected, and since there are no edges between distinct 
connected components $G_i$ and $G_j$ of $G[u]$, we have that for each $i\in \{1,...,k\}$, there is at least one edge between 
a vertex of $G_i$ and a vertex in $V\backslash V[u]$. But by Equation \ref{equation:CarvingWidth}, 
this implies that the width of $(T,\gamma)$ is at least $k$, contradicting in this way the assumption that the width 
of $(T,\gamma)$ is $w$. $\square$
\end{proof}

Let $G$ be a connected graph with $n$ vertices. We say that an ordering $v_1v_2...v_n$ of the vertices of $G$ 
is a breadth first traversal in $G$ if for every 
$i,j \in \{1,..,n\}$ with $i<j$, we have that the distance from $v_1$ to $v_i$ is at most the distance from $v_1$ to $v_j$. 
We note that if $v_1v_2...v_n$ is a breadth first traversal, then for each $k\in \{2,...,n\}$, there is an edge
connecting $v_k$ to some vertex in $\{v_1,...,v_{k-1}\}$.

\begin{proposition}
\label{proposition:BipartiteWedges}
Let $G=(V,E)$ be a connected graph with $n$ vertices and $w$ edges. There is a contractive carving decomposition of $G$ of height 
$n-1$ and width at most $w$.
\end{proposition}
\begin{proof}
Since $G$ has $w$ edges, any rooted carving decomposition of $G$ has width at most $w$. Thus we just need to 
show that some rooted carving decomposition $(T,\gamma)$ of $G$ is contractive. We let $T$ be the unique binary 
tree with $n$ leaves and height $n-1$. In other words, $T$ has $n-1$ internal nodes, and each of these nodes has 
a child that is a leaf. Additionally, if $u$ is the internal node of $T$ farthest away from the root then 
both children of $u$ are leaves. Now we define the function $\gamma$ which is a bijection from the leaves of $T$
to the vertices of $G$. Let $v_1v_2...v_{n}$ be a breadth first traversal of the vertices of  
$G$. Then for each $k\in \{2,...,n\}$, there is an edge connecting $v_k$ to some 
vertex in $\{v_1,...,v_{k-1}\}$. 
Let $u^*$ be one of the two leaves of $T$ at distance $n-1$ from the root. 
We set $\gamma(u^*)=v_1$. Now for each leaf $u\neq u^*$ we set $\gamma(u)=v_k$
if and only if the distance from $u$ to the root is equal to $n-k+1$ (see Fig. \ref{figure:TrivialContractiveCarvingDecomposition}).

\begin{figure}[h]
\centering
\includegraphics[scale=0.35]{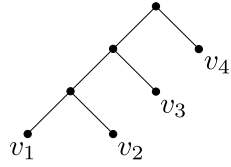}
\caption{The contractive carving decomposition corresponding to a breadth first traversal $v_1v_2v_3v_4$ of a 
graph with $4$ vertices.} 
\label{figure:TrivialContractiveCarvingDecomposition}
\end{figure}

 We claim that $(T,\gamma)$ is contractive.
Let $u$ be an internal node of $T$. If both children of $u$ are leaves, then both of them are at distance $n-1$ from the root. 
By definition, one of these leaves is $u^*$, which is labeled with $v_1$. The other leaf is labeled with $v_2$. 
Since $v_1v_2...v_n$ is a breadth first traversal of the vertices of $G$, there is an edge connecting $v_1$ and $v_2$.
Now let $u$ be 
an internal node for which both children are at distance $n-k+1$ from the root for $k > 2$. Then one of the children of $u$, say $u.r$,
is a labeled with $v_k$ and the other leaf, say $u.l$ is such that $V[u.l]=\{v_1,...,v_{k-1}\}$. Again, since the 
sequence $v_1v_2...v_{n}$ is a breadth first traversal of $G$, we have that there is at least one edge from $V[u.r]=\{v_k\}$ to 
$V[u.l]$.
This shows that $(T,\gamma)$ is contractive. $\square$ 
\end{proof} 

Let $G=(V,E)$ be a connected graph and let $G_1,...,G_k$ be induced subgraphs 
of $G$ such that $G_i=(V_i,E_i)$ for $i\in \{1,...,k\}$, $V_i\cap V_j = \emptyset$ for $i\neq j$, 
and $V = \bigcup_{i}V_i$. We denote by $\graphcal(G,G_1,...,G_k)$ the 
graph with vertex set $\mathcal{V} = \{G_1,...,G_k\}$, and whose edge set $\mathcal{E}$ has one edge $(G_i,G_j)$ for 
each edge of $G$ with one endpoint in some vertex of $G_i$ and another endpoint in some vertex of $G_j$. We note that 
there may be multiple edges between two induced subgraphs $G_i$ and $G_j$. 
Let $(T_{\graphcal},\gamma_{\graphcal})$ be a contractive carving decomposition of $G$ of width at most $w$ and, for each $i\in \{1,...,k\}$,
let $(T_i,\gamma_i)$ be a contractive carving decomposition of width at most $w$ of $G_i$. We denote by 
\begin{equation}
\label{equation:CompositionCarving}
(T,\gamma) = (T_{\graphcal},\gamma_{\graphcal})\oplus [(T_1,\gamma_1), ... , (T_k,\gamma_k)]
\end{equation}
the carving decomposition of $G$ that is obtained by identifying the root of each $(T_i,\gamma_i)$ with the leaf $u$ of $T_{\graphcal}$ for which 
$\gamma_{\graphcal}(u)=G_i$ (see Fig. \ref{figure:SumOfCarvingDecompositions}). 
It is immediate to check that $(T,\gamma)$ is a contractive carving decomposition of $G$ of width at most $w$.

\begin{figure}[h]
\centering
\includegraphics[scale=0.40]{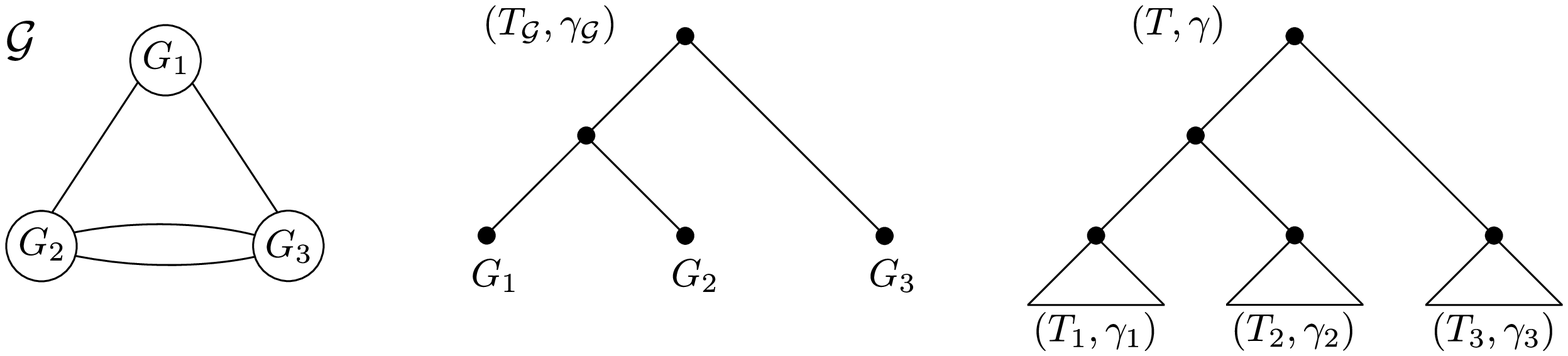}
\caption{Left: The graph $\graphcal = \graphcal(G,G_1,G_2,G_3)$, whose vertices $G_1,G_2,G_3$ are induced subgraphs of a graph $G$. 
Middle: $(T_{\graphcal},\gamma_{\graphcal})$ is a contractive carving decomposition of ${\graphcal = (G,G_1,G_2,G_3)}$.
Right: For each $i$, $(T_i,\gamma_i)$ is a contractive carving decomposition of $G_i$. $(T,\gamma)$ is a 
contractive carvind decomposition of $G$ obtained by identifying, for each $i$, the 
root of $(T_i,\gamma_i)$ with the leaf of $(T_{\graphcal},\gamma_{\graphcal})$ labeled with $G_i$. If 
$(T_{\graphcal},\gamma_{\cal})$ and $(T_i,\gamma_i)$ have width at most $w$, then $(T,\gamma)$ has width at most $w$.} 
\label{figure:SumOfCarvingDecompositions}
\end{figure}

\begin{observation}
\label{observation:VerticesEdgesPlusOne}
Let $G$ be a connected graph with $w$ edges. Then $G$ has at most $w+1$ vertices. 
\end{observation}
\begin{proof}
The proof is by induction on the number of edges. In the base case, $G$ has a unique edge, 
and therefore the observation holds trivially. Now assume that, for each $w\geq 2$, the observation holds for every 
graph with $w-1$ edges, and let $G$ be a connected graph with $w$ edges. Every such graph $G$ can be
obtained by adding an edge $e$ to a connected graph $G'$ with $w-1$ edges. By the induction hypothesis, 
$G'$ has at most $w$ vertices. Since $G$ is connected, at least one of the vertices in $e$ belongs to $G'$. Therefore $G$ has at most $w+1$ vertices. 
\end{proof}

\paragraph{\bf Proof of Lemma \ref{lemma:GoodCarving}}
Let $G=(V,E)$ be a connected graph, and let $(T,\gamma)$ be a rooted carving decomposition of $G$ of width $w$ and height $h$. 
We will construct a {\em contractive} carving decomposition $(T',\gamma')$ of $G$ of width at most $w$ and depth at most $w\cdot h$. 
Let $u$ be a node of $T$. By proposition \ref{proposition:DisjointSetsSubtree}, 
the graph $G[u]$ has at most $w$ connected components. Let $G_{u,1},...,G_{u,r}$ for $r\leq w$ be the connected components of $G[u]$.
Let $\height(u)$ denote the height of node $u$ in $T$. 

\begin{claim}
\label{claim:Contractive}
For each node $u$ of $T$ and each connected component $G_{u,i}$ of $G[u]$, 
there exists a {\em contractive} carving decomposition $(T_{u,i},\gamma_{u,i})$ of $G_{u,i}$
such that the height of $(T_{u,i},\gamma_{u,i})$ is at most $w\cdot \mathit{height}(u)$. 
\end{claim}

We note that Claim $\ref{claim:Contractive}$ implies Lemma \ref{lemma:GoodCarving}, since 
if $u$ is the root of $T$ and $G$ is connected, then the graph $G[u]$ has a single connected component, 
which is $G$ itself. The proof of Claim \ref{claim:Contractive} is by induction on the height of the node $u$ in $T$. 
In the base case, $u$ is a leaf of $T$. In this case $G[u]$ consists of a unique connected 
component $G_{u,1}$ which is the vertex $\gamma(u)$ of $V$ labeling $u$. And therefore the carving decomposition 
$(T_{u,1},\gamma_{u,1})$ consists of a unique node $u'$ labeled with $\gamma_{u,1}(u')=\gamma(u)$.

Now assume that Claim \ref{claim:Contractive} is true for every node of $T$ of height at most $h$, 
and let $u$ be a node of height $h+1$. By the induction hypothesis, each connected component $G_{u.l,i}$
of $G[u.l]$ has a contractive carving decomposition $(T_{u.l,i},\gamma_{u.l,i})$ of height at most $w\cdot \height(u.l)$. 
Analogously, each connected component $G_{u.r,j}$ of $G[u.r]$ has a contractive carving decomposition 
$(T_{u.r,j},\gamma_{u.r,j})$ of height at most $w\cdot \mathit{height}(u.l)$. 

Let $G_{u,j}$ be a connected component of $G[u]$. Then the set of vertices of $G_{u,j}$ is the union of the vertex sets of 
some connected components of $G[u.l]$ (say $G_{u.l,1},...G_{u.l,p}$), 
and some connected components of $G[u.r]$ (say $G_{u.r,1},...,G_{u.r,q}$).
Note that all edges of $G_{u,j}$ that do not belong to some of these components, 
must connect some vertex of $G_{u.l,i}$ to some vertex of $G_{u.r,i'}$ for 
some $i\in \{1,...,p\}$ and some $i'\in \{1,...,q\}$. 
But since the carving decomposition $(T,\gamma)$  of the graph $G$ has  width at most $w$, 
there can be at most $w$ such new edges. In other words, the graph 
$\graphcal = \graphcal(G_{u,j}, G_{u.l,1},...G_{u.l,p},G_{u.r,1},...,G_{u.r,q})$ has at most $w$ edges. 
Since $\graphcal$ is connected, by Observation \ref{observation:VerticesEdgesPlusOne} we have that $\graphcal$ has 
at most $w+1$ vertices. By Proposition \ref{proposition:BipartiteWedges}, $\graphcal$ has 
a contractive carving decomposition $(T_{\graphcal},\gamma_{\graphcal})$ of height at most $w$ and width at most $w$. 
Therefore the carving decomposition 
\begin{equation*}
(T_{u,i}) = (T_{\graphcal},\gamma_{\graphcal})\oplus [(T_{u.l,1},\gamma_{u.l,1}),...,(T_{u.l,p},\gamma_{u.l,p}), 
(T_{u.r,1},\gamma_{u.r,1}),...,(T_{u.r,q},\gamma_{u.r,q})]
\end{equation*} 
is contractive and has width at most $w$. Since $(T_{\graphcal},\gamma_{\graphcal})$ has height at most $w$, and 
by assumption, each $(T_{u.l,i},\gamma_{u.l,i})$ and $(T_{u.r,j},\gamma_{u.r,j})$ has height at most $w\cdot h$, 
we have that the height of $(T_{u,i},\gamma_{u,i})$ is at most $w\cdot h + w = w\cdot (h+1)$. This 
proves Claim \ref{claim:Contractive}, and therefore also Lemma \ref{lemma:GoodCarving}, by letting $u$ be the root of $T$. $\square$

\section{Tensor Networks}
\label{section:TensorNetwork}

In this section we will redefine the well known notion of tensor network in function of abstract networks. 
Within this formalism, a tensor network is a pair $(\abstractnetwork,\lambda)$ where $\abstractnetwork$
is an abstract network, and $\lambda$ is a function that associates a tensor $\lambda(\indexset)$
of rank $|\indexset|$ 
with each index set $\indexset\in \abstractnetwork$. 
We believe that defining tensor networks in this way has the advantage of separating the algorithmic aspects of tensor networks 
from their quantum aspects. Additionally, the formalism of abstract networks will also be used in Section \ref{section:FeasibilityTensorNetworks} 
to introduce the notion of {\em feasibility tensor networks} which will be used to address the problem of 
approximating the maximum acceptance probability of quantum circuits with uninitialized inputs. 

Let $\Pi(d) =\{\ket{b_1}\bra{b_2}\;|\; b_1,b_2\in \{0,...,d-1\}\}$. A $d$-state tensor 
with index set ${\indexset =\{i_1,...,i_k\}}$ is an array $\atensor$ consisting of $|\Pi(d)|^{k}=d^{2k}$ complex numbers. 
The entries $$\atensor(\sigma_{i_1},...,\sigma_{i_k})$$ of $\atensor$ are indexed by a sequence of variables 
$\sigma_{i_1},...,\sigma_{i_k}$, each of which ranges over the set $\Pi(d)$. We note that if $\indexset=\emptyset$
then a tensor with index set $\indexset$ is simply a complex number $\atensor(\_)$. 
If $\atensor$ is a tensor with index set $\indexset$ then we let $\rank(\atensor)=|\indexset|$ be the {\em rank} of $\atensor$.  
We denote by $\tensors(d,\indexset)$ the set of all $d$-state tensors with index set $\indexset$ and by 
$\tensors(d)=\bigcup_{\indexset \subseteq \N} \tensors(d,\indexset)$ the set of all $d$-state tensors. 

\begin{definition}[Tensor Network]
\label{definition:TensorNetwork}
A {\em tensor network} is a pair $(\abstractnetwork,\lambda)$ where $\abstractnetwork$ is an abstract network and $\lambda$ is a function that associates with each 
index set $\indexset \in \abstractnetwork$, a tensor $\lambda(\indexset)\in \tensors(d,\indexset)$. 
\end{definition}

A tensor network $(\abstractnetwork,\lambda)$ is connected if $\abstractnetwork$ is connected. In this work we will 
only be interested in connected tensor networks. 
An important operation involving tensors is the operation of tensor contraction. 
If $\atensor$ is a tensor with index set ${\indexset=\{i_1,...,i_k,l_1,...,l_r\}}$ and $\atensor'$ is a tensor with index set 
$\indexset'=\{j_1,...,j_{k'},l_1,...,l_r\}$ then the contraction of $\atensor$ and $\atensor'$ gives rise to the tensor 
$\tensorContraction(\atensor,\atensor')$ with index set ${\indexset\oplus \indexset' = \{i_1,...,i_k, j_1,...,j_{k'}\}}$ 
where each entry $\tensorContraction(\atensor,\atensor')(\sigma_{i_1},...,\sigma_{i_k},\sigma_{j_1},...,\sigma_{j_{k'}})$
is defined as

\begin{equation}
\sum_{\sigma_{l_1},...,\sigma_{l_r}\in \Pi(d)} \atensor(\sigma_{i_1},...,\sigma_{i_k},\sigma_{l_1},...,\sigma_{l_{r}})\cdot \atensor'(\sigma_{j_1},...,\sigma_{j_{k'}}, \sigma_{l_1},...,\sigma_{l_r}).
\end{equation}

If $(\abstractnetwork,\lambda)$ is a tensor network and $\indexset_1,\indexset_2$ is a pair of contractible sets in $\abstractnetwork$ 
then we say that the tensor network $(\abstractnetwork',\lambda')$ is obtained from $(\abstractnetwork,\lambda)$ by the 
contraction of $\indexset_1$ and $\indexset_2$ if 
$\abstractnetwork' = (\abstractnetwork\backslash \{\indexset_1,\indexset_2\}) \cup \{\indexset_1\oplus \indexset_2\}$, 
and if $\lambda'$ satisfies the following conditions.  
\begin{enumerate}
	\item $\lambda'(\indexset_1\oplus \indexset_2)  = \tensorContraction(\lambda(\indexset_1),\lambda(\indexset_2))$.
	\item $\lambda'(\indexset)  = \lambda(\indexset)$ for each $\indexset \in \abstractnetwork'\backslash \{\indexset_1\oplus \indexset_2\}$. 
\end{enumerate}

Any connected tensor network with $m$ index sets can be contracted $m-1$ times. The result of this contraction process 
is a tensor network $([\,\emptyset\,],\lambda_0)$ with 
a unique index set, namely $\emptyset$, which is labeled with a rank-$0$ tensor $\lambda_0(\emptyset)$ 
(that is to say, a complex number). The value of $(\abstractnetwork,\lambda)$, denoted by 
$\valuetensornetwork(\abstractnetwork,\lambda)$, is defined as the absolute value of $\lambda_0(\emptyset)$. 
More precisely, $\valuetensornetwork(\abstractnetwork,\lambda) = |\lambda_0(\emptyset)|$. 
We observe that the value of a tensor network is well defined, since it does not depend on the order 
in which the tensors of the network are contracted. 

\subsection{Mapping Quantum Circuits with Initialized Inputs to Tensor Networks}

One of the main reasons behind the popularity of tensor networks is the fact that they can be used to simulate quantum circuits.
First, we note that both density operators and quantum gates can be naturally regarded as tensors. 
If $\rho$ is a density operator acting on $d$-dimensional qudits indexed by $\indexset=\{i_1,...,i_k\}$, then 
the tensor $\mathbold{\rho}$ associated with $\rho$ is defined as 

\begin{equation}
\mathbold{\rho}(\sigma_{i_1},...,\sigma_{i_k})= \trace\left(\rho\cdot [\sigma_{i_1}^{\dagger} \otimes ...\otimes \sigma_{i_k}^{\dagger}] \right). 
\end{equation}

If $Q$ is a quantum gate with inputs indexed by $\indexset=\{i_1,...,i_k\}$ and outputs indexed by $\indexset'=\{j_1,...,j_l\}$ where 
$\indexset\cap \indexset'=\emptyset$, then the tensor $\mathbold{Q}$ associated with $Q$ is defined as 

\begin{equation}
\mathbold{Q}(\sigma_{i_1},...,\sigma_{i_k},\sigma_{j_1},...,\sigma_{j_l}) = \trace\left( Q \cdot 
[\sigma_{i_1}\otimes ...\otimes \sigma_{i_k}] \cdot [\sigma_{j_1}^{\dagger}\otimes ...\otimes \sigma_{j_l}^{\dagger}] \right).
\end{equation}

In the sequel, we will not distinguish between gates or density matrices and their associated tensors. 
If $C=(V,E,\vertexlabelingfunction,\edgelabelingfunction)$ is a quantum circuit in which all inputs are initialized, 
then the tensor network $(\abstractnetwork_C,\lambda_{C})$ associated with $C$ is obtained as follows. For each vertex $v\in V$, 
let $\indexset(v)$ be the index set consisting of all integers labeling edges of $C$ which are incident with $v$. Then 
we add $\indexset(v)$ to $\abstractnetwork_C$ and set $\lambda_C(\indexset(v))$ to be the tensor associated 
with the gate $\vertexlabelingfunction(v)$ of $C$. We say that $\abstractnetwork_{C}$ is the abstract network associated with 
$C$.  
The following proposition, which is well known in tensor-network theory (see \cite{MarkovShi2008} for a proof), 
establishes a close correspondence between the value of tensor networks and the acceptance probability of quantum circuits. 

\begin{proposition}
\label{proposition:Measurement}
Let $C$ be a quantum circuit with $n$ inputs initialized with the state $\ket{y}$ for some $y\in \{0,...,d-1\}^n$. 
Then $\valuetensornetwork(\abstractnetwork_C,\lambda_C) = Pr(C,\ket{y})$. 
\end{proposition}

In other words, $\valuetensornetwork(\abstractnetwork_C,\lambda_C)$ is the acceptance probability of $C$. 

\subsection{Computing the Value of a Tensor Network}
\label{subsection:ComputingValueTensorNetwork}

The process of computing the value $\valuetensornetwork(\abstractnetwork,\lambda)$ of a tensor network $(\abstractnetwork,\lambda)$ 
is known as simulation. Given a contraction tree $(T,\iota)$ of rank $r$ for $\abstractnetwork$,
the following definition can be used to compute $\valuetensornetwork(\abstractnetwork,\lambda)$ 
in time $d^{O(r)}\cdot |\abstractnetwork|^{O(1)}$.

\begin{definition}[Tensor Network Simulation]
\label{definition:TensorNetworkSimulation}
Let $(\abstractnetwork,\lambda)$ be a tensor network and $(T,\iota)$ be a contraction tree for $\abstractnetwork$. 
A simulation of $(\abstractnetwork,\lambda)$ on $(T,\iota)$ is a function $\simulation:\nodes(T)\rightarrow \tensors(d)$ 
satisfying the following conditions:
\begin{enumerate}
	\item For each leaf $u$ of $T$, $\simulation(u) = \lambda(\iota(u))$.
	\item For each internal node $u$ of $T$, $\simulation(u)= \tensorContraction(\simulation(u.l),\simulation(u.r))$. 
\end{enumerate}
\end{definition}

Note that if $u$ is the root of a contraction tree, then $\iota(u)=\emptyset$. In this case the tensor $\simulation(u)$ is 
a rank-0 tensor (that is, a complex number) which is obtained by contracting all tensors in $(\abstractnetwork,\lambda)$. 
This implies that $|\simulation(u)| = \valuetensornetwork(\abstractnetwork,\lambda)$.
If the contraction tree $(T,\iota)$ has rank $r$, then for each node $u$ of $T$ we have that $|\iota(u)|\leq r$. In other 
words, for each $u\in \nodes(T)$, the tensor $\simulation(u)$ has rank at most $r$, and for this reason 
$\simulation(u)$ can be represented by $d^{2r}$ complex numbers.
In this way, the simulation  $\simulation$ can be inductively constructed in time $d^{O(r)}\cdot |\abstractnetwork|^{O(1)}$, 
and therefore, $\valuetensornetwork(\abstractnetwork,\lambda)$ can be computed in time $d^{O(r)}\cdot |\abstractnetwork|^{O(1)}$. 

By Theorem \ref{theorem:GoodContractionTree}, if the graph $\graph(\abstractnetwork)$ of an abstract network 
$\abstractnetwork$ has treewidth $t$ and maximum degree $\Delta$, then one can construct in polynomial time a contraction tree for 
$\abstractnetwork$ of rank $O(\Delta\cdot t)$. Therefore, Definition \ref{definition:TensorNetworkSimulation} can be used to compute  
the value $\valuetensornetwork(\abstractnetwork,\lambda)$ of a tensor network $(\abstractnetwork,\lambda)$ in time 
${d^{O(\Delta\cdot t)}\cdot |\abstractnetwork|^{O(1)}}$. 
Now let $\universalgates$ be a fixed finite universal set of gates, and let $C$ be a quantum circuit of treewidth $t$, 
whose inputs are initialized with a basis state $\ket{y}$, and whose gates
are drawn from $\universalgates$. Let $(\abstractnetwork_C,\lambda_C)$ be the tensor network associated with $C$.
By Proposition \ref{proposition:Measurement}, $\valuetensornetwork(\abstractnetwork_C,\lambda_C) = \mathit{Pr}(C,\ket{y})$. 
Therefore, the simulation algorithm described above can be used to compute the acceptance probability of $C$ in 
time $d^{O(\Delta(\universalgates)\cdot t)}\cdot|C|^{O(1)}$, where $|C|$ is the number of vertices of $C$.
We note that this algorithm has the same asymptotic time complexity as the original contraction algorithm for tensor 
networks devised in \cite{MarkovShi2008}, although our contraction technique based 
on Theorem \ref{theorem:GoodContractionTree} is different from that employed in \cite{MarkovShi2008}. 

We observe that the fact that the contraction trees constructed in Theorem \ref{theorem:GoodContractionTree} 
have logarithmic height is not relevant for the time complexity of the simulation algorithm described above.
Nevertheless, as we will see in Section \ref{section:ApproximatingFeasibility}, contraction trees of logarithmic
height will be essential when devising a polynomial time algorithm for the problem of approximating 
the maximum acceptance probability of constant-treewidth quantum circuits with uninitialized inputs.
Even though it is possible to extract contraction trees of constant rank from the contraction sequences 
defined in \cite{MarkovShi2008}, the contraction trees obtained in this way are not guaranteed to have 
logarithmic height. Therefore, our contraction algorithm cannot be directly replaced by the contraction algorithm 
devised in \cite{MarkovShi2008} when addressing the satisfiability of constant-treewidth quantum circuits.

\section{Feasibility Tensor Networks}
\label{section:FeasibilityTensorNetworks}

In Section \ref{section:TensorNetwork} we defined tensor networks in terms of abstract 
networks and showed how contraction trees can be used to address the problem of computing the 
value of a tensor network. In this section we will use abstract networks to introduce {\em feasibility tensor networks}. 
We will then proceed to show 
that feasibility tensor networks can be used to address the problem of computing an assignment that maximizes the acceptance 
probability of quantum circuits with uninitialized inputs. 

\begin{definition}[Feasibility Tensor Network]
A {\em feasibility tensor network} is a pair $(\abstractnetwork, \Lambda)$ where $\abstractnetwork$ is an abstract network 
and $\Lambda:\abstractnetwork\rightarrow 2^{\tensors(d)}$ is a function that associates with each index set $\indexset\in \abstractnetwork$ a finite 
set of tensors $\Lambda(\indexset)\subseteq \tensors(d,\indexset)$. 
\end{definition}

Note that the only difference between tensor networks and feasibility tensor networks is that while in the former we 
associate a tensor with each index set, in the latter we associate a set of tensors with each index set. 
If $(\abstractnetwork,\Lambda)$ is a feasibility tensor network, then an {\em initialization} of $(\abstractnetwork,\Lambda)$ 
is a function $\lambda:\abstractnetwork\rightarrow \tensors(d)$ such that $\lambda(\indexset) \in \Lambda(\indexset)$ for 
each index set $\indexset\in \abstractnetwork$. Intuitively, an initialization $\lambda$ chooses one tensor $\lambda(\indexset)$
from each set of tensors $\Lambda(\indexset)$. 
Observe that for each such an initialization $\lambda$, the pair $(\abstractnetwork,\lambda)$ is a tensor network as defined 
in Section \ref{section:TensorNetwork}. The value of a feasibility tensor network is defined as 

\begin{equation}
\label{equation:ValueFeasibilityTensorNetwork}
\valuefeasibilitytensornetwork(\abstractnetwork,\Lambda) = \max\{\valuetensornetwork(\abstractnetwork,\lambda)\;|\;\lambda \mbox{ is an initialization of $(\abstractnetwork,\Lambda)$}\}.
\end{equation}

Below we show that 
the problem of finding an assignment that maximizes the acceptance probability of a quantum circuit 
with uninitialized inputs can be reduced to the problem of computing an initialization of 
maximum value for a feasibility tensor network. 
Therefore, the problem of computing the value of a feasibility tensor network is $\mathrm{QCMA}$ hard. The 
conversion from quantum circuits with uninitialized inputs to feasibility tensor networks 
goes as follows: Each uninitialized input $v$ corresponds to 
the set of tensors $\{\ket{0}\bra{0},\ket{1}\bra{1},...,\ket{d-1}\bra{d-1}\}$. Intuitively, this set of tensors 
consists of all possible values that can be used to initialize $v$. On the other hand, each input vertex $v$ which
is already initialized with a density matrix $\ket{i}\bra{i}$ corresponds to the singleton set 
$\{\ket{i}\bra{i}\}$. Finally, each gate $g$ of the circuit corresponds to the singleton set 
$\{\atensor\}$. We formalize this construction in Definition \ref{definition:CircuitsUninitializedFeasibility}.  

\begin{definition}[From Quantum Circuits to Feasibility Tensor Networks]
\label{definition:CircuitsUninitializedFeasibility}
Let ${C=(V,E,\vertexlabelingfunction,\edgelabelingfunction)}$ be a quantum circuit in which some of the inputs are uninitialized.
The feasibility tensor network associated with $C$ is denoted by $(\abstractnetwork_C,\Lambda_C)$, where 
${\abstractnetwork_{C}=\{\indexset(v)\; | \; v\in V\}}$ is the abstract network associated with $C$, 
and $\Lambda_C$ is such that for each $v\in V$, 
\begin{equation}
\label{equation:LAMBDA}
\Lambda_C(\indexset(v)) = \left\{ 
\begin{array}{l}
\{\ket{0}\bra{0}, \ket{1}\bra{1},...,\ket{d-1}\bra{d-1}\} \mbox{ if $v$ is an uninitialized input,}\\
\\
\{\vertexlabelingfunction(v)\} \mbox{ otherwise.} \\
\end{array}\right.
\end{equation}
\end{definition}

Now let $\lambda$ be an initialization of the feasibility tensor network $(\abstractnetwork_C,\Lambda_C)$. 
Then the tensor network $(\abstractnetwork_C, \lambda)$ is precisely the tensor network 
associated with the circuit $C$ in which the inputs are initialized with the state

\begin{equation}
\label{equation:AssignmentFromInitialization}
\ket{y_{\lambda}}= \otimes_{v\in \inputvertices(C)} \lambda(\indexset(v)).
\end{equation} 

In other words, 
$\valuetensornetwork(\abstractnetwork_{C},\lambda) = \mathit{Pr}(C,\ket{y_{\lambda}})$.  
Therefore, we have the following observation.

\begin{observation} 
\label{observation:QuantumCircuitsToTensorNetworks}
For each quantum circuit $C$, the value $\valuefeasibilitytensornetwork(\abstractnetwork_{C},\Lambda_C)$ of the 
feasibility tensor network $(\abstractnetwork_C,\Lambda_C)$ associated with $C$ is equal to the maximum acceptance probability $\mathit{Pr}^{cl}(C)$ of $C$.
\end{observation}

\section{Approximating the Value of a Feasibility Tensor Network}
\label{section:ApproximatingFeasibility}

In this section we will devise an algorithm that, when given a feasibility tensor network $(\abstractnetwork,\Lambda)$ and 
a real number $\delta\in (0,1)$ as input, can be used both to approximate the value $\valuefeasibilitytensornetwork(\abstractnetwork,\Lambda)$ up to a $\delta$ additive 
factor, and to construct an initialization $\lambda$ such that 
$|\valuefeasibilitytensornetwork(\abstractnetwork,\Lambda) - \valuetensornetwork(\abstractnetwork,\lambda)|\leq \delta$. In particular, 
our algorithm runs in polynomial time if we are given a contraction tree for $\abstractnetwork$ of constant rank and logarithmic height.
As we saw in Section \ref{section:GoodContractionTree} if the graph $\graph(\abstractnetwork)$ associated with 
$\abstractnetwork$ has constant treewidth and constant maximum degree, then a contraction tree with these properties 
can be efficiently constructed using Theorem \ref{theorem:GoodContractionTree}.

\subsection{Tensor $\varepsilon$-Nets}
\label{subsection:TensorEpsilonNets}

We start by defining suitable notions of norm and distance for tensors.
If $\atensor$ is a tensor with index-set $\indexset = \{i_1,...,i_k\}$, then the 
$\lang_{\infty}$ norm of $\atensor$ is defined as
 
\begin{equation}
\label{equation:}
\|\atensor\|  = \max_{\sigma_{i_1}...\sigma_{i_k}} |\atensor(\sigma_{i_1},...,\sigma_{i_k})|, 
\end{equation}

where for each $j\in \{1,...,k\}$, $\sigma_{i_j}$ ranges over the set $\Pi(d)$, and 
$|\atensor(\sigma_{i_1},...,\sigma_{i_k})|$ is the absolute value of the entry $\atensor(\sigma_{i_1},...,\sigma_{i_k})$ of $\atensor$. 
Having the definition of norm of a tensor in hands, we define the distance between two tensors $\atensor$ and $\atensor'$ 
as ${|\atensor - \atensor'\|}$. The next step consists in defining a suitable notion of $\varepsilon$-net of tensors.

\begin{definition}[Tensor $(d,\epsilon,\indexset)$-Net]
Let $\indexset$ be an index set, $d\in \N$, and $\varepsilon\in \R$ with $0< \varepsilon < 1$. 
A tensor $(d,\varepsilon,\indexset)$-net is a set $\tensorSpace(d,\varepsilon,\indexset)$ of $d$-state tensors with index 
set $\indexset$ such that for each $\atensor$ in $\tensorSpace(d,\indexset)$, there exists $\atensor'\in \tensorSpace(d,\varepsilon,\indexset)$ 
with $\|\atensor-\atensor'\| \leq \varepsilon$. 
\end{definition}

It is straightforward to construct a $(d,\varepsilon,\indexset)$-net for tensors. We just need to consider the set of all $d$-state tensors with index set 
$\indexset$ in which each entry is a complex number of the form $a+b\cdot i$ for $-1\leq a,b\leq 1$ and $a,b$ integer multiples of $\varepsilon/2$. 
We observe that we do not need to assume that the tensors in our $(d,\varepsilon,\indexset)$-net correspond to physically realizable 
operations. Our  approximation algorithm does not need this assumption. Since a $d$-state tensor over the index set $\indexset$ has $d^{2|\indexset|}$ 
entries, we have the following proposition upper bounding the size of a tensor $(d,\varepsilon,\indexset)$-net.   

\begin{proposition}
\label{proposition:EpsilonNet}
For each index set $\indexset$, each $d\in \N$ and each $\varepsilon\in \R$ with ${0<\varepsilon <1}$, 
one can construct a tensor $(d,\varepsilon,\indexset)$-net $\tensors(d,\varepsilon,\indexset)$ with
at most 
$(1/\varepsilon)^{\exp(O(|\indexset|\log d))}$ tensors. 
\end{proposition}

If $\atensor$ is a tensor in $\tensors(d,\indexset)$ , then we let $\truncation_{\varepsilon}(\atensor)$ be an arbitrary 
tensor $\atensor'$ in $\tensors(d,\varepsilon,\indexset)$ such that $\|\atensor-\atensor'\|\leq \varepsilon$.
Going further, if $\feasibilitySet$ is a set of tensors then we let $$\truncation_{\varepsilon}(\feasibilitySet) = \{\truncation_{\varepsilon}(\atensor)\;|\;\atensor\in\feasibilitySet\}$$
be the truncated version of $\feasibilitySet$.

\subsection{Approximation Algorithm}
\label{subsection:ApproximationAlgorithm}

In this subsection we will address the problem of $\delta$-approximating the value of feasibility tensor 
networks and the problem of finding $\delta$-optimal initializations for feasibility tensor networks. 
First, we define the notion of contraction for pairs of sets of tensors. 
Let $\indexset,\indexset'$ be a pair of index sets with  $\indexset \cap \indexset'\neq \emptyset$. 
Let $\tensorset\subseteq \tensors(d,\indexset)$  be a finite set of tensors with index set $\indexset$ and 
$\tensorset' \subseteq \tensors(d,\indexset')$ be a finite set of tensors with index set $\indexset'$. The contraction 
of $\tensorset$ and $\tensorset'$ is defined as 

\begin{equation}
\tensorContraction(\tensorset,\tensorset') = \{\tensorContraction(\atensor,\atensor')\; |\; \atensor\in \tensorset, \atensor'\in \tensorset'\}. 
\end{equation}

Subsequently, we define a notion of simulation for feasibility tensor networks that is analog to our definition of 
simulation for tensor networks introduced in Subsection \ref{subsection:ComputingValueTensorNetwork}. The simulation of 
a feasibility tensor network $(\abstractnetwork,\Lambda)$ on a contraction tree $(T,\iota)$ is a function 
$\truncatedfeasibilitysimulation$ that associates a set of tensors with each node of $T$. 
First, with each leaf $u$ of $T$ such that $\iota(u)=\indexset$, we associate the set of tensors 
$\truncatedfeasibilitysimulation(\indexset) = \Lambda(\indexset)$.
Then, with each internal node $u$ of $T$, we associate the set of tensors 
$\truncatedfeasibilitysimulation(u) = \truncation_{\varepsilon}(\tensorContraction(\truncatedfeasibilitysimulation(u.l),\truncatedfeasibilitysimulation(u.r)))$.
We note that the truncation is necessary to keep the size of each set from growing exponentially as the contractions take place.
This construction is given more formally in Definition \ref{definition:TruncatedSimulation} below. 

\begin{definition}[Feasibility Tensor Network Simulation]
\label{definition:TruncatedSimulation}
Let $(\abstractnetwork,\Lambda)$ be a feasibility tensor network and $(T,\iota)$ be a contraction tree for $\abstractnetwork$. An 
{$\varepsilon\mbox{-simulation}$} of $(\abstractnetwork,\Lambda)$ on $(T,\iota)$ is a function 
$\truncatedfeasibilitysimulation:N\rightarrow 2^{\tensors(d,\varepsilon)}$ 
satisfying the following properties:
\begin{enumerate}
	\itemsep0.3em
	\item For each leaf $u$ of $T$, $\truncatedfeasibilitysimulation(u) = \Lambda(\iota(u))$,
	\item For each internal node $u$ of $T$, 
	$\truncatedfeasibilitysimulation(u)= \truncation_{\varepsilon}(\tensorContraction(\truncatedfeasibilitysimulation(u.l),\truncatedfeasibilitysimulation(u.r)))$. 
\end{enumerate}
\end{definition}

Intuitively, an $\varepsilon$-simulation $\truncatedfeasibilitysimulation$ is a function that keeps track of all ways of 
simulating tensor networks $(\abstractnetwork,\lambda)$ where $\lambda$ is an arbitrary initialization of $(\abstractnetwork,\Lambda)$. 
In particular, if $u$ is the root of $(T,\iota)$ then $u$ is labeled with a set $\truncatedfeasibilitysimulation(u)$ of 
complex numbers. For each such complex number $a$, there exists an initialization $\lambda$ of 
$(\abstractnetwork,\Lambda)$ such that $|a|$ is an approximation of $\valuetensornetwork(N,\lambda)$. Conversely, 
for each initialization $\lambda$ of $(\abstractnetwork,\Lambda)$, there exists some number $a\in \truncatedfeasibilitysimulation(u)$
such that $|a|$ approximates $\valuetensornetwork(\abstractnetwork,\lambda)$. 
Therefore, the maximum absolute value $\alpha$ of a complex number in $\truncatedfeasibilitysimulation(u)$ 
is an approximation of $\valuefeasibilitytensornetwork(\abstractnetwork,\Lambda)$. 
An actual initialization $\lambda$ of $(\abstractnetwork,\Lambda)$
such that $\valuetensornetwork(\abstractnetwork,\lambda) = \alpha \approx \valuefeasibilitytensornetwork(\abstractnetwork,\Lambda)$
can be found by backtracking. Theorem \ref{theorem:TensorNetworkSatisfiability}
below, which will be proved in Subsection \ref{proof:theorem:TensorNetworkSatisfiability}, 
establishes an upper bound for the time complexity and for the error of the approximation scheme described above. 
The error of such process depends exponentially on the height of the contraction tree, while the time complexity depends exponentially on the rank 
of the contraction tree.

\begin{theorem}[Feasibility Tensor Network Satisfiability]
\label{theorem:TensorNetworkSatisfiability}
Let $(\abstractnetwork,\Lambda)$ be a feasibility tensor network, $(T,\iota)$ be a
contraction tree for $\abstractnetwork$ of rank $r$ and height $h$, and $\varepsilon$ be a real number with 
$0<\varepsilon <1$. 
\begin{enumerate}
	\item \label{TensorNetworkSatisfiability-itemOne} One can compute a number $\alpha$ such that 
		$|\alpha-\valuefeasibilitytensornetwork(\abstractnetwork,\Lambda)|\leq \varepsilon\cdot (3d^{2r}+1)^{h}$ in 
		time  $|\abstractnetwork|\cdot \varepsilon^{-\exp(O(r\log d))}$.
	\item  \label{TensorNetworkSatisfiability-itemTwo} One can construct in time 
		$|\abstractnetwork|\cdot \varepsilon^{-\exp(O(r \log d))}$ 
		an initialization $\lambda$ of $(\abstractnetwork,\Lambda)$ such that 
		$$|\valuetensornetwork(\abstractnetwork,\lambda) - \valuefeasibilitytensornetwork(\abstractnetwork,\Lambda)|\leq \varepsilon\cdot (3d^{2r}+1)^h.$$
		
\end{enumerate}
\end{theorem}

We note that to efficiently compute $\alpha$ and $\lambda$ in Theorem \ref{theorem:TensorNetworkSatisfiability} above, we 
need to have in hands a contraction tree for $\abstractnetwork$ whose rank is bounded by a constant, 
and whose height is bounded by $O(\log |\abstractnetwork|)$. The next theorem (Theorem \ref{theorem:ApproximationFeasibility}) 
states that approximately optimal initializations of feasibility tensor networks of constant treewidth and constant maximum 
degree can be computed in polynomial time. Note that in this case, the existence of a contraction tree of constant rank and 
logarithmic height is guaranteed by Theorem \ref{theorem:GoodContractionTree}.

\begin{theorem}
\label{theorem:ApproximationFeasibility}
Let $(\abstractnetwork,\Lambda)$ be a feasibility tensor network such that the graph $\graph(\abstractnetwork)$ has 
treewidth $t$ and maximum degree $\Delta$. For each $\delta$ with ${1/\poly(n)< \delta< 1}$, one 
can compute in time $(|\abstractnetwork|/ \delta)^{\exp(O(\Delta\cdot t\cdot \log d))}$ an initialization $\lambda$ of $(\abstractnetwork,\Lambda)$
such that $$|\valuetensornetwork(\abstractnetwork,\lambda)-\valuefeasibilitytensornetwork(\abstractnetwork,\Lambda)|\leq \delta.$$
\end{theorem}
\begin{proof}
By Theorem \ref{theorem:GoodContractionTree}, we can construct a contraction tree for $\abstractnetwork$ of 
rank ${r=O(\Delta \cdot t)}$ and height $h=O(\Delta\cdot t \cdot \log|\abstractnetwork|)$. By 
Theorem \ref{theorem:TensorNetworkSatisfiability}.\ref{TensorNetworkSatisfiability-itemTwo}, 
we can compute in time ${|\abstractnetwork|^{O(1)}\cdot \varepsilon^{-\exp(O(r \log d))}}$ 
an initialization $\lambda$ of ${(\abstractnetwork,\Lambda)}$ such that 
$${|\valuetensornetwork(\abstractnetwork,\lambda)-\valuefeasibilitytensornetwork(\abstractnetwork,\Lambda)|\leq 
\varepsilon\cdot 2^{O(r\cdot h\cdot \log d)}}.$$ Therefore, by setting 
$\varepsilon = \delta/|\abstractnetwork|^{O(\Delta^2\cdot t^2\cdot \log d)}$, 
we can compute an  initialization $\lambda$ for 
$(\abstractnetwork,\Lambda)$ such that
$|\valuetensornetwork(\abstractnetwork,\lambda)-\valuefeasibilitytensornetwork(\abstractnetwork,\Lambda)|\leq \delta$
in time $(|\abstractnetwork|/\delta)^{\exp(O(\Delta\cdot t\cdot \log d))}$. 
 $\square$
\end{proof}

Since the problem of computing a $\delta$-optimal initialization of a quantum circuit can be 
reduced to the problem of computing a $\delta$-optimal initialization of a feasibility tensor 
network, our main theorem (Theorem \ref{theorem:CircuitSatisfiability}) follows from 
${\mbox{Theorem \ref{theorem:ApproximationFeasibility}}}$.

\paragraph{\bf Proof of Theorem \ref{theorem:CircuitSatisfiability}}
Let $C$ be a quantum circuit with $n$ uninitialized inputs, treewidth $t$,
and $\mathit{poly}(n)$ gates drawn from a finite universal set of gates $\universalgates$. 
Let $(\abstractnetwork_{C},\Lambda_{C})$ be the feasibility tensor network associated with 
$C$ according to Definition \ref{definition:CircuitsUninitializedFeasibility}. 
Then the graph $\graph(\abstractnetwork_{C})$ has treewidth $t$, and maximum degree
$\Delta(\universalgates)$, where $\Delta(\universalgates)$ is the maximum number of 
inputs and outputs of a gate in $\universalgates$. 
Additionally, 
$\mathit{Pr}^{cl}(C) = \valuefeasibilitytensornetwork(\abstractnetwork_{C},\Lambda_{C})$ and, 
by Equation \ref{equation:AssignmentFromInitialization}, each 
initialization $\lambda$ of $(\abstractnetwork_{C},\Lambda_{C})$ corresponds to 
a classical initialization $\ket{y_{\lambda}}$ of $C$  such that 
${\mathit{Pr}^{cl}(C,\ket{y_{\lambda}}) = \valuetensornetwork(\abstractnetwork_{C},\lambda)}$. Using Theorem \ref{theorem:ApproximationFeasibility}, we can compute 
in time $(n/\delta)^{\exp(O(\Delta(\universalgates)\cdot t\cdot \log d))}$ an initialization $\lambda$ of 
$(\abstractnetwork_{C},\Lambda_{C})$ such that 
$$|\valuetensornetwork(\abstractnetwork_{C},\lambda) - 
\valuefeasibilitytensornetwork(\abstractnetwork_{C},\Lambda_{C})|\leq \delta.$$
Since $\valuefeasibilitytensornetwork(\abstractnetwork_{C},\Lambda_C) = \mathit{Pr}^{cl}(C)$,
 we have  $|\mathit{Pr}(C,\ket{y_{\lambda}}) - \mathit{Pr}^{cl}(C)|\leq \delta$. 
Finally, since $\Delta(\universalgates)$ and $d$ are constants, the time complexity of this 
construction can be simplified to $(n/\delta)^{O(t)}$.
$\square$

\subsection{Proof of Theorem \ref{theorem:TensorNetworkSatisfiability}}
\label{proof:theorem:TensorNetworkSatisfiability}

In this subsection we will prove Theorem \ref{theorem:TensorNetworkSatisfiability}. In the proof we will 
devise an algorithm to $\delta$-approximate the value of a feasibility tensor network 
and to compute an initialization which achieves a near optimal value. We start by stating a couple 
of auxiliary lemmas. In particular, the following lemma establishes an upper bound for the distance 
between the contraction of two given tensors $\atensor_1$ and $\atensor_2$ and the contraction of 
approximations $\atensor_1'$ and $\atensor_2'$ of $\atensor_1$ and $\atensor_2$ respectively.

\begin{lemma}
\label{lemma:ErrorContraction}
Let $\atensor_1$ and $\atensor_1'$ be tensors with index set $\indexset_1$ and let 
$\atensor_2$ and $\atensor_2'$ be tensors with index set $\indexset_2$, where $\indexset_1\cap \indexset_2\neq \emptyset$. 
Let ${\|\atensor_1 - \atensor_1'\| \leq \varepsilon}$ and  ${\|\atensor_2- \atensor_2'\| \leq \varepsilon}$.
Then $$\|\tensorContraction(\atensor_1,\atensor_2) - \tensorContraction(\atensor_1',\atensor_2')\| \leq \varepsilon \cdot 3d^{2\cdot |\indexset_1\cap \indexset_2|}.$$
\end{lemma}
\begin{proof}
Let $\indexset_1 = \{i_1,...i_k, l_1,...,l_r\}$ and $\indexset_2 = \{j_1,...,j_{k'},l_1,...,l_r\}$ be 
index sets. 
Let $\atensor_1' =\atensor_1 + \mathbold{e}_1$  and $\atensor_2' =\atensor_2 + \mathbold{e}_2$ where 
$\mathbold{e}_1$ and $\mathbold{e}_2$ are offset tensors. Since $\|\atensor_1-\atensor_1'\|\leq \varepsilon$ and 
$\|\atensor_2-\atensor_2'\|\leq \varepsilon$,
we have that $\|\mathbold{e}_1\| \leq \varepsilon$ and $\|\mathbold{e}_2\| \leq \varepsilon$.
Let $\sigma = (\sigma_{i_1},...,\sigma_{i_k},\sigma_{l_1},...\sigma_{l_{r}})$ and $\sigma' = (\sigma_{j_1},...,\sigma_{j_{k'}}, \sigma_{l_1},...,\sigma_{l_r})$. 
Then we have that 

$$
\tensorContraction(\atensor_1',\atensor_2')(\sigma_{i_1},...,\sigma_{i_k},\sigma_{j_1},...,\sigma_{j_{k'}}) =
\sum_{\sigma_{l_1}...\sigma_{l_r}\in \Pi(d)} (\atensor_1(\sigma) + \mathbold{e}_1(\sigma))
(\atensor_2(\sigma') + \mathbold{e}_2(\sigma')).$$

By reorganizing the right hand side,
$\tensorContraction(\atensor_1',\atensor_2')(\sigma_{i_1},...,\sigma_{i_{k}},\sigma_{j_1},...,\sigma_{j_{k'}})$ 
is equal to

$$\sum_{\sigma_{l_1},...,\sigma_{l_r}\in \Pi(d)} \atensor_1(\sigma)\atensor_2(\sigma') + 
\sum_{\sigma_{l_1},...,\sigma_{l_r}\in \Pi(d)}  [\mathbold{e}_1(\sigma) \atensor_2(\sigma') + 
\mathbold{e}_2(\sigma') \atensor_1(\sigma) + \mathbold{e}_1(\sigma)\mathbold{e}_2(\sigma') ]$$

The first term of this sum is simply $\tensorContraction(\atensor_1,\atensor_2)(\sigma_{i_1},...\sigma_{i_k},\sigma_{j_1},...,\sigma_{j_{k'}})$. The 
second term can be simplified by noting that ${|\mathbold{e}_1(\sigma)|\leq \varepsilon}$, ${|\mathbold{e}_2(\sigma')|\leq \varepsilon}$, 
${|\atensor_1(\sigma)| \leq 1}$, ${|\atensor_2(\sigma')| \leq 1}$ and that 
${|\mathbold{e}_1(\sigma)\mathbold{e}_2(\sigma')|\leq \varepsilon}$. Additionally, 
since ${\indexset_1\cap \indexset_2 = \{l_1,...,l_r\}}$, there are at most  
$d^{2|\indexset_1\cap \indexset_2|}$ sequences 
of the form $\sigma_{l_1},...,\sigma_{l_r}$ where $\sigma_{l_i}\in \Pi(d)$ for each $i\in \{1,...,r\}$. 
Therefore, $\tensorContraction(\atensor_1',\atensor_2')(\sigma_{i_1},....,\sigma_{i_k},\sigma_{j_1},...,\sigma_{j_{k'}})$  is equal to 

$$
\tensorContraction(\atensor_1,\atensor_2)(\sigma_{i_1},....,\sigma_{i_k},\sigma_{j_1},...,\sigma_{j_{k'}})  \pm 
\varepsilon \cdot (3d^{2|\indexset_1\cap \indexset_2|}).$$

$\square$ 
\end{proof}

We observe that for each  three tensors $\atensor_1$, $\atensor_2$ and $\atensor_3$, if ${\|\atensor_1-\atensor_2\|\leq \varepsilon}$ 
and ${\|\atensor_2-\atensor_3\| \leq \varepsilon'}$ then $\|\atensor_1 - \atensor_3\| \leq \varepsilon + \varepsilon'$. Using 
this observation, the following lemma is a consequence of Lemma \ref{lemma:ErrorContraction}.

\begin{lemma}
\label{lemma:ErrorContractionTruncated}
Let $\atensor_1$ and $\atensor_1'$ be tensors  with index set $\indexset_1$, and 
$\atensor_2$ and $\atensor_2'$ be tensors  with index set $\indexset_2$ where 
$\indexset_1\cap \indexset_2 \neq \emptyset$, $|\indexset_1|\leq r$ and $|\indexset_2|\leq r$. 
Let 
 ${\|\atensor_1 - \atensor_2\| \leq \varepsilon\cdot (3d^{2\cdot r}+1)^h}$ and
${\|\atensor_1'- \atensor_2'\| \leq  \varepsilon\cdot (3d^{2\cdot r}+1)^h}$.
Then $$\|\truncation_{\varepsilon}(\tensorContraction(\atensor_1,\atensor_2))-
\tensorContraction(\atensor_1',\atensor_2')\|  \leq \varepsilon \cdot (3d^{2\cdot r}+1)^{h+1}.$$
\end{lemma}
\begin{proof}
By our definition of truncation, 
$$\|\truncation_{\varepsilon}(\tensorContraction(\atensor_1,\atensor_2)) - \tensorContraction(\atensor_1,\atensor_2)\| \leq \varepsilon.$$
Since $|\indexset_1|\leq r$ and $|\indexset_2|\leq r$, we have $|\indexset_1\cap \indexset_2|\leq r$. Therefore, 
by Lemma \ref{lemma:ErrorContraction}, 
$$\|\tensorContraction(\atensor_1,\atensor_2) - \tensorContraction(\atensor_1',\atensor_2')\| \leq 
\left[\varepsilon \cdot (3d^{2\cdot r} +1)^h \right] \cdot 3d^{2\cdot r}.$$
This implies that 
$$
\begin{array}{lcl}
\|\truncation_{\varepsilon}(\tensorContraction(\atensor_1,\atensor_2)) - \tensorContraction(\atensor_1',\atensor_2')\| &  \leq &  
\varepsilon \cdot (3d^{2\cdot r}+1)^h \cdot 3d^{2\cdot r}\, +\, \varepsilon \\
\\
		&  \leq &  \varepsilon \cdot (3d^{2\cdot r}+1)^{h+1}.
\end{array}
$$
$\square$

\end{proof}

Next, we define the notion of {\em partial simulation}. Recall that if 
$T$ is a tree, then we denote by $T[u]$ the subtree of $T$ rooted at 
$u$.

\begin{definition}[Partial Simulation]
\label{definition:PartialSimulation}
Let $(\abstractnetwork,\Lambda)$ be a feasibility tensor network, $(T,\iota)$ be 
a contraction tree for $\abstractnetwork$ and $u$ be a node of $T$. 
A partial simulation of $(\abstractnetwork,\Lambda)$ rooted at $u$ is a function 
$\evaluation_u:\nodes(T[u])\rightarrow \tensors(d)$ satisfying the following conditions. 
\begin{enumerate}
	\item For each leaf $u'$ of $T[u]$, $\evaluation_u(u') \in \Lambda(u')$.
	\item For each internal node $u'$ of $T[u]$, $\evaluation_u(u') = \tensorContraction(\evaluation_u(u'.l),\evaluation_u(u'.r))$
\end{enumerate}
\end{definition}

Intuitively, a partial simulation $\evaluation_u:\nodes(T[u])\rightarrow \tensors(d)$ of $(\abstractnetwork,\Lambda)$ 
can be obtained by the following process. First, we consider some initialization 
$\lambda$ of $(\abstractnetwork,\Lambda)$. Subsequently, we construct the simulation $\evaluation:\nodes(T)\rightarrow \tensors(d)$ of 
the tensor network $(\abstractnetwork,\lambda)$ on the contraction tree $(T,\iota)$. Finally, 
we restrict $\evaluation$ to the nodes of $T[u]$. In other words, we set $\evaluation_u = \evaluation|_{\nodes(T[u])}$. 
In particular, we note that if $u$ is the root of $T$, then $\evaluation_u = \evaluation$. 
The next lemma establishes an upper bound for the error propagation 
during the process of constructing an $\varepsilon$-simulation for a feasibility tensor network.

\begin{lemma}
\label{lemma:ErrorSimulation}
Let $(\abstractnetwork,\Lambda)$ be a feasibility tensor network, $(T,\iota)$ be a contraction tree for $\abstractnetwork$ of
rank $r$ and $\truncatedfeasibilitysimulation:N\rightarrow 2^{\tensors(d,\varepsilon)}$ be an $\varepsilon$-simulation 
of $(\abstractnetwork,\Lambda)$ on $(T,\iota)$. 
Finally, let $u$ be a node of $T$.
\newline
\begin{enumerate}
\setlength\itemsep{1em}
	\item \label{item:ErrorSimulationOne} For each partial simulation ${\evaluation_u:\nodes(T[u])\rightarrow \tensors(d)}$
	of $(\abstractnetwork,\Lambda)$ rooted at $u$, 	there is a tensor ${\atensor\in \truncatedfeasibilitysimulation(u)}$ 
		such that $\|\evaluation_u(u) - \atensor\| \leq \varepsilon\cdot (3d^{2\cdot r}+1)^{\mathit{height}(u)}$.
 	\item \label{item:ErrorSimulationTwo} For each tensor $\atensor\in \truncatedfeasibilitysimulation(u)$, there is a 
		partial simulation ${\evaluation_u:\nodes(T[u])\rightarrow \tensors(d)}$ of $(\abstractnetwork,\Lambda)$ rooted
		at $u$ such that $\|\evaluation_u(u) - \atensor\| \leq \varepsilon\cdot (3d^{2\cdot r}+1)^{\mathit{height}(u)}$.
\end{enumerate} 
\end{lemma}
\begin{proof}
The proofs of both Lemma \ref{lemma:ErrorSimulation}.\ref{item:ErrorSimulationOne} and 
Lemma \ref{lemma:ErrorSimulation}.\ref{item:ErrorSimulationTwo} follow by induction on the height of $u$. 
First, we note that since the contraction tree $(T,\iota)$ has rank $r$, we have that
${|\indexset(u.l)\cap\indexset(u.r)|\leq r}$ for each internal node $u$ of $T$. Thus,
all tensors associated with nodes of $T$ have rank at most $r$. Now we proceed with the 
proof by induction. In the base case, $u$ is a leaf and therefore, $\nodes(T[u])=\{u\}$.
In this case, for each partial simulation 
$\evaluation_u:\{u\}\rightarrow \tensors(d)$, the tensor $\evaluation_u(u)$ belongs to 
$\Lambda(\iota(u)) = \truncatedfeasibilitysimulation(u)$
by Definition \ref{definition:PartialSimulation}. Conversely, for each tensor $g\in \truncatedfeasibilitysimulation(u)$, 
the function $\evaluation_u:\{u\}\rightarrow \tensors(d)$ obtained by setting $\evaluation_u(u)=g$ is 
a valid partial simulation.  
Now, suppose that the lemma is valid for every node of height at most $h$ and let $u$ be a node of height $h+1$. 
\begin{enumerate}
	\item Let $\evaluation_u:\nodes(T[u])\rightarrow \tensors(d)$ be a partial simulation. 
		Let $\evaluation_{u.l}=\evaluation_{u}|_{T[u.l]}$ and 
		$\evaluation_{u.r} = \evaluation_{u}|_{T[u.r]}$ be the restrictions 
		of $\evaluation_u$ to the nodes of the subtrees $T[u.l]$ and $T[u.r]$ respectively. Note 
		that $\evaluation_{u.l}(u.l) = \evaluation_u(u.l)$ and $\evaluation_{u.r}(u.r)=\evaluation_u(u.r)$, 
		and therefore, by Definition \ref{definition:PartialSimulation}, 
		$$\evaluation_{u}(u) = \tensorContraction(\evaluation_{u.l}(u.l),\evaluation_{u.r}(u.r)).$$ 
	 	By the induction hypothesis, there exist tensors 
		$g.l\in \truncatedfeasibilitysimulation(u.l)$ and $g.r\in \truncatedfeasibilitysimulation(u.r)$ 
		such that
		$$\|\evaluation_{u.l}(u.l) - \atensor.l\| \leq \varepsilon \cdot (3d^{2r}+1)^h$$
		and 
		$$\|\evaluation_{u.r}(u.r) - \atensor.r\| \leq \varepsilon \cdot (3d^{2r}+1)^h.$$ 
		By Definition \ref{definition:TruncatedSimulation}, the tensor 
		$\atensor = \truncation_{\varepsilon}(\tensorContraction(\atensor.l,\atensor.r))$ 
		belongs to $\truncatedfeasibilitysimulation(u)$. Finally, by Lemma \ref{lemma:ErrorContractionTruncated}, 
		$$\|\evaluation_u(u) - \atensor\| \leq \varepsilon\cdot (3d^2r+1)^{h+1}.$$
 
 	\item Let $\atensor$ be a tensor in $\truncatedfeasibilitysimulation(u)$.  By Definition \ref{definition:TruncatedSimulation},
		there exist tensors $\atensor.l\in \truncatedfeasibilitysimulation(u.l)$ and 
		$\atensor.r \in \truncatedfeasibilitysimulation(u.r)$
		such that $\atensor = \truncation_{\varepsilon}(\tensorContraction(\atensor.l,\atensor.r))$. 
		By the induction hypothesis, there exist partial simulations $\evaluation_{u.l}:\nodes(T[u.l])\rightarrow \tensors(d)$
		and $\evaluation_{u.r}:\nodes(T[u.r])\rightarrow \tensors(d)$ such that
		$$\|\evaluation_{u.l}(u.l)-\atensor.l\| \leq \varepsilon\cdot (3d^{2r}+1)^{h}$$ 
		and
		$$\|\evaluation_{u.r}(u.r)-\atensor.r\| \leq \varepsilon\cdot (3d^{2r}+1)^{h}.$$
		Now let $\evaluation_{u}:\nodes(T[u])\rightarrow \tensors(d)$ be the partial simulation
		that extends $\evaluation_{u.l}$ and $\evaluation_{u.r}$ by one node. More precisely,
		restricting $\evaluation_u$ to the nodes of $T[u.l]$ yields $\evaluation_{u.l}$, restricting $\evaluation_u$ to the 
		nodes of $T[u.r]$ yields $\evaluation_{u.r}$, and the tensor associated by $\evaluation_u$ with the 
		node $u$ is the contraction of the tensors associated by $\evaluation_{u.l}$ and $\evaluation_{u.r}$ with 
		the nodes $u.l$ and $u.r$ respectively. Formally, $\evaluation_u$ is defined by setting
		$$\evaluation_{u}|_{T[u.r]} = \evaluation_{u.r}, \hspace{0.5cm}
		\evaluation_{u}|_{T[u.l]} = \evaluation_{u.l},\hspace{0.5cm}\mbox{and}$$ 
 		$$\evaluation_u(u) = \tensorContraction(\evaluation_{u.l}(u.l),\evaluation_{u.r}(u.r)).$$
		Therefore, by Lemma \ref{lemma:ErrorContractionTruncated},
		$$\|\evaluation_u(u)-\atensor\| \leq \varepsilon\cdot (3d^2r+1)^{h+1}.$$
		$\square$ 
\end{enumerate}
\end{proof}

The next lemma, which is a consequence of Lemma \ref{lemma:ErrorSimulation}, states that 
if $u_0$ is the root of $T$ then the maximum absolute value of a complex number in 
$\truncatedfeasibilitysimulation(u_0)$ is at most $(3d^{2\cdot r} +1)^h$ apart 
from the value $\valuefeasibilitytensornetwork(\abstractnetwork,\Lambda)$ of the feasibility 
tensor network $(\abstractnetwork,\Lambda)$.

\begin{lemma}
\label{lemma:TruncationWellDefined}
Let $(\abstractnetwork,\Lambda)$ be a feasibility tensor network and $(T,\iota)$ be a contraction tree for $\abstractnetwork$ of
rank $r$ and height $h$. Let $\truncatedfeasibilitysimulation:N\rightarrow 2^{\tensors(d,\varepsilon)}$ be an $\varepsilon$-simulation 
of $(\abstractnetwork,\Lambda)$ on $(T,\iota)$. If $u_0$ is the root of~ $T$ and $\alpha$ is the largest absolute value 
of a complex number in $\truncatedfeasibilitysimulation(u_0)$, then 
$|\alpha - \valuefeasibilitytensornetwork(\abstractnetwork,\Lambda)|\leq (3d^2r+1)^h$. 
\end{lemma}
\begin{proof}

First, we note that since $u_0$ is the root of $T$, all elements of $\truncatedfeasibilitysimulation(u_0)$ 
are complex numbers, i.e., rank-0 tensors.

\begin{claim} For each initialization $\lambda$ of $(\abstractnetwork,\Lambda)$ there 
is a complex number in $\truncatedfeasibilitysimulation(u_0)$ such that
 $|\,\valuetensornetwork(\abstractnetwork,\lambda) - |g|\,|\leq (3d^{2\cdot r} +1)^h$.
Conversely, for each $g\in \truncatedfeasibilitysimulation(u_0)$ there is 
an initialization $\lambda$ of $(\abstractnetwork,\Lambda)$ such that 
$|\,\valuetensornetwork(\abstractnetwork,\lambda) - |g|\,|\leq (3d^{2\cdot r} +1)^h$.

\end{claim}

Recall that $\valuefeasibilitytensornetwork(\abstractnetwork,\Lambda)$ is defined as 
$\max_{\lambda} \valuetensornetwork(\abstractnetwork,\lambda)$ 
where $\lambda$ ranges over all initializations of $(\abstractnetwork,\Lambda)$.
Therefore, the claim stated above implies that if $\alpha$ is the maximum absolute 
value of a complex number in $\truncatedfeasibilitysimulation(u_0)$, 
then $$|\alpha - \valuefeasibilitytensornetwork(\abstractnetwork,\Lambda)|\leq (3d^{2\cdot r}+1)^h.$$ 

Now we proceed to prove our claim. First, 
let $\lambda$ be an initialization of $(\abstractnetwork,\Lambda)$. Then there exists a partial simulation 
$\evaluation_{u_0}$ of $(\abstractnetwork,\Lambda)$ rooted at $u_0$ such that ${\evaluation_{u_0} = \evaluation}$, 
where $\evaluation$ is the simulation of the tensor network $(\abstractnetwork,\lambda)$ constructed as in 
Definition \ref{definition:TensorNetworkSimulation}. Note that $\valuetensornetwork(\abstractnetwork,\lambda)= |\evaluation_{u_0}(u_0)|$, 
since $u_0$ is the root of $T$. By Lemma \ref{lemma:ErrorSimulation}.$i$, there exists a complex number
$g\in \truncatedfeasibilitysimulation(u_0)$ such that $|\evaluation_{u_0}(u_0) - g|\leq (3d^{2\cdot r}+1)^h$. 
Using the fact that $|\,|x|-|y|\,|\leq |x-y|$ for every pair of complex numbers $x$ and $y$, we have that 
$|\,\valuetensornetwork(\abstractnetwork,\lambda)- |g|\,|\leq (3d^{2\cdot r}+1)^h$. 

Conversely, let $g$ be a complex number in $\truncatedfeasibilitysimulation(u_0)$. By Lemma \ref{lemma:ErrorSimulation}.$ii$, there 
exists a partial simulation $\lambda_{u_0}$ of $(\abstractnetwork,\Lambda)$ rooted at $u_0$ such that 
${|\lambda_{u_0}(u_0) - g|\leq (3d^{2\cdot r} + 1)^h}$. Since $u_0$ is the root of $T$, there exists an initialization 
$\lambda$ of $(\abstractnetwork,\Lambda)$ such that ${\evaluation_{u_0} = \evaluation}$, where $\evaluation$ is the 
simulation of the tensor network $(\abstractnetwork,\lambda)$ on $(T,\iota)$ constructed according to Definition 
\ref{definition:TensorNetworkSimulation}. Note again that $\valuetensornetwork(\abstractnetwork,\lambda) = |\evaluation_{u_0}(u_0)|$. 
Therefore using the fact that $|\,|x|-|y|\,| \leq |x-y|$ for every pair of complex numbers $x$ and $y$, we have that 
$|\,\valuetensornetwork(\abstractnetwork,\lambda)- |g|\,|\leq (3d^{2\cdot r}+1)^h$. This proves the claim. $\square$ 

\end{proof}

Finally, we are in a position to prove Theorem \ref{theorem:TensorNetworkSatisfiability}. 

\paragraph{\bf Proof of Theorem \ref{theorem:TensorNetworkSatisfiability}}\hspace{0.5cm}\\

\ref{TensorNetworkSatisfiability-itemOne}) Let $(\abstractnetwork,\Lambda)$ be a feasibility tensor 
network and $(T,\iota)$ be a contraction tree for $\abstractnetwork$ of rank $r$ and height $h$. 
Let $\truncatedfeasibilitysimulation:N\rightarrow 2^{\tensors(d)}$ be the 
$\varepsilon$-simulation $(\abstractnetwork,\Lambda)$ on $(T,\iota)$ constructed 
according to Definition \ref{definition:TruncatedSimulation}.   
Since $T$ is a binary tree with $|\abstractnetwork|$ leaves, the total number of nodes in $T$ is $2|\abstractnetwork|-1$. 
Additionally, since $(T,\iota)$ has rank $r$, for each node $u$ of $T$, the set $\truncatedfeasibilitysimulation(u)$ has at most 
$|\tensors(d,\varepsilon,r)|=\varepsilon^{-\exp(O(r\cdot \log d))}$ tensors. Therefore, $\truncatedfeasibilitysimulation$ 
can be constructed in time $|\abstractnetwork|\cdot \varepsilon^{-\exp(O(r\cdot \log d))}$.
Now let $u_0$ be the root of $T$, and let $\alpha$ be the largest absolute value of a complex number 
in $\truncatedfeasibilitysimulation(u_0)$. Then by Lemma \ref{lemma:TruncationWellDefined},
$|\alpha - \valuefeasibilitytensornetwork(\abstractnetwork,\Lambda)|\leq (3d^{2\cdot r}+1)^h$.
Therefore, after having constructed $\truncatedfeasibilitysimulation$ we just need to output $\alpha$.

\ref{TensorNetworkSatisfiability-itemTwo}) Let $u_0$ be the root of $T$ and let 
$g$ be the complex number in $\truncatedfeasibilitysimulation(u_0)$ with largest 
absolute value $\alpha$. We construct a simulation $\simulation:\nodes(T)\rightarrow \tensors(d)$ of 
$(\abstractnetwork,\Lambda)$ on $(T,\iota)$ as follows. First, we set $\evaluation(u_0) = g$. 
Now for each internal node $u$ for which $\simulation(u)$ has already been determined, 
let $u_l$ and $u_r$ be respectively the left and right children of $u$. 
Then we set $\simulation(u.l)=g_l$ and $\simulation(u.r)=g_r$ where $g_l$ and $g_r$ are respectively tensors 
in $\truncatedfeasibilitysimulation(u.l)$ and $\truncatedfeasibilitysimulation(u.r)$ for which 
$\simulation(u)=\truncation_{\varepsilon}(g_l,g_r)$. We proceed in this way until we 
have determined $\simulation$ on all leaves of $T$. The searched initialization $\lambda$ 
is then obtained by considering the tensors associated by $\simulation$ with the leaves 
of $T$. In other words, for each leaf $u$ labeled with the index set $\iota(u)$, 
we set $\lambda(\iota(u)) = \simulation(u)$. Since $\simulation(u_0) = g$, we have that 
$\valuetensornetwork(\abstractnetwork,\lambda)= \alpha$. Since $T$ has $2|\abstractnetwork|-1$ nodes,
once we are given $\truncatedfeasibilitysimulation$, the construction of the initialization
$\lambda$ takes time $O(|\abstractnetwork|)$. Since $\truncatedfeasibilitysimulation$ can 
be constructed in time $|\abstractnetwork|\cdot \varepsilon^{-\exp(O(r\cdot \log d))}$, 
the overall time complexity to construct $\lambda$ is $|\abstractnetwork|\cdot \varepsilon^{-\exp(O(r\cdot \log d))}$. 
$\square$

\section{Classical Witnesses vs Quantum Verifiers of Logarithmic Width}
\label{section:ProofLogarithmicTreewidthNP}

In this Section we will prove Theorem \ref{theorem:LogarithmicTreewidthNP}, which 
states that for any constant $\delta$ with $0<\delta < 1$, and any quantum circuit $C$ with $n$ uninitialized inputs, 
$\mathit{poly}(n)$ gates, and online-width $O(\log n)$,  
it is $\mathrm{NP}$ complete to determine whether ${Pr^{cl}(C)=1}$ or whether ${Pr^{cl}(C)\leq \delta}$. 
We note that as an implication of Theorem \ref{theorem:LogarithmicTreewidthNP}, we have that 
${\mathrm{QCMA}[\onlinecutwidth,O(\log n)] = \mathrm{NP}}.$
Indeed, from the proof of Theorem \ref{theorem:LogarithmicTreewidthNP}, it will also follow that
$\mathrm{QCMA}[\treewidth,O(\log n)] = \mathrm{NP}$.

\paragraph{\bf Membership in $\mathrm{NP}$:}
As mentioned in the introduction, Markov and Shi's simulation algorithm \cite{MarkovShi2008} computes 
(with polynomially many bits of precision) the acceptance probability of a quantum circuit of treewidth 
$t$ in deterministic time $2^{O(t)}\cdot |C|^{O(1)}$. 
Let $C$ be a  quantum circuit with $n$ uninitialized inputs, $\poly(n)$ gates, and 
treewidth $O(\log n)$. Given a classical assignment 
$y\in \{0,...,d-1\}^{n}$ we can use Markov and Shi's algorithm to compute $\mathit{Pr}(C,\ket{y})$ 
in time $n^{O(1)}$. We accept $y$ as a suitable witness if $\mathit{Pr}(C,\ket{y})=1$ and 
we reject $y$ if $\mathit{Pr}(C,\ket{y})\leq \delta$. This shows that the problem of computing a
$\delta$-optimal classical witness for a quantum circuit of logarithmic treewidth is in $\mathrm{NP}$.
Since, by Equation \ref{equation:ComparisonMeasures}, the treewidth of any circuit is 
upper bounded by its online width, we have that the problem of computing a 
$\delta$-optimal classical witness for a quantum circuit of logarithmic online-width is also in $\mathrm{NP}$.
In summary, we have just shown that 
$$\mathrm{QCMA}[\onlinecutwidth,O(\log n)] \subseteq \mathrm{QCMA}[\treewidth,O(\log n)] \subseteq \mathrm{NP}.$$

\paragraph{\bf $\mathrm{NP}$-hardness:}
To show that the problem of finding a $\delta$-optimal classical assignment for a
quantum circuit of logarithmic online-width is $\mathrm{NP}$-hard, we will first consider 
a probabilistic verifier of logarithmic online-width for the $3$-SAT problem which
has inverse polynomial soundness. Subsequently we will show that the soundness can be 
amplified to a constant by increasing the online-width of the original verifier by 
a logarithmic additive factor. Clearly, such hardness result for probabilistic circuits
of logarithmic online width implies the same hardness result for quantum circuits of 
logarithmic online width. Also, since the treewidth of a circuit is always upper-bounded 
by its online-width, the following sequence of inclusions is implied by the 
$\mathrm{NP}$-hardness proof provided in this section. 

$$\mathrm{NP} \subseteq \mathrm{QCMA}[\onlinecutwidth,O(\log n)] \subseteq \mathrm{QCMA}[\treewidth,O(\log n)].$$

\begin{proposition}[Folklore]
\label{proposition:VerificationCircuit}
Let  $F(x_1,x_2,...,x_n)$ be a $3$-SAT formula with $n$ variables $x_1...x_n$ and $m=poly(n)$ clauses $W_1,W_2,...,W_m$. 
There is a probabilistic circuit $C_F$ of online-width $O(\log n)$ such that the following conditions are 
satisfied. 
\begin{enumerate}
	\item (Completeness) $\;$ If $F$ is satisfiable then ${Pr^{\mathit{cl}}(C_F) = 1}$.
	\item (Soundness) $\;$ If $F$ is not satisfiable then $Pr^{\mathit{cl}}(C_F)\leq 1-1/poly(n)$. 
\end{enumerate}
\end{proposition}
\begin{proof}
The circuit $C_F$ simulates the execution of an online verifier that does the following: 
First it chooses a number $r$ from the set $\{1,...,m\}$ at random. Then, when reading the 
witness provided by Merlin, it ignores all bits which assign values to variables that do 
not belong to $W_r$, and keeps only those bits which correspond to variables used by $W_r$. 
Subsequently it verifies whether $W_r$ is evaluated to true. If the formula $F$ has a satisfying 
assignment, then the prover can always convince the verifier with probability $1$, since in 
this case the restriction of a satisfying assignment to the variables occurring in $W_r$ will 
cause $W_r$ to evaluate to $1$. Otherwise, if $F$
has no satisfying assignment, then no matter what assignment is provided by the prover,  there will 
be at least one clause of $F$ which will not be satisfied. Thus the verifier will accept with 
probability at most $1-1/m$. $\square$
\end{proof}

\newcommand{\maj}{\mathrm{MAJ}}
\newcommand{\adder}{\mathrm{ADD}}
\newcommand{\comparison}{\mathrm{COMP}}
It remains to show that the same hardness 
result holds if the soundness in Proposition \ref{proposition:VerificationCircuit} is constant, instead of 
inverse polynomial. This can be done by a standard probability amplification argument: First we 
create polynomially many copies of the original verifier, each with access to independent random bits, 
and then we consider the majority vote of the answer of all verifiers.  We claim that if the original verifier 
$C$ has online-width $w$, then the overall amplified circuit has online-width $w+O(\log n)$.

Let $q = n^{O(1)}$, and let $\adder$ be an adder with $c+1$ input bits and $c$ output bits where $c=\lceil \log q \rceil$.
One of the inputs of $\adder$ is a control bit. The other $c$ inputs of $\adder$ form a counter register. The circuit 
$\adder$ increases the value of its counter register by one if the control bit is set to $1$, and does nothing 
with the input otherwise. Let $\comparison$ be a circuit with $c$ input bits and one output bit. The circuit
$\comparison$ accepts if the value represented by its inputs is greater than $\lceil \frac{q}{2} \rceil$, and  
rejects otherwise. Finally let $\maj(x_1,...,x_q)$ be a circuit that accepts an input $x_1x_2...x_q$ if 
$\sum_{i=1}^q x_i \geq \lceil \frac{q}{2} \rceil$ and rejects otherwise. 
Then $\maj(x_1,...,x_q)$ can be implemented as a sequence of sub-circuits $\adder_1,\adder_2,...,\adder_q,\comparison$,
where for each $i\in \{1,...,q\}$, $\adder_i$ is an adder with control bit $x_i$, 
for each $i\in \{1,...,q-1\}$ the outputs of $\adder_i$ are identified with the inputs at the counter register 
of $\adder_{i+1}$, the outputs of $\adder_q$ are identified with the inputs of $\comparison$, 
and the inputs at the counter register of $\adder_1$ are set to $0$ (Fig. \ref{figure:Majority-Vote}). 

Now let $C_F$ be the probabilistic circuit of online width $w= O(\log n)$ considered in  
Proposition \ref{proposition:VerificationCircuit}. Let $C_1,...,C_q$ be $q$ independent copies of the circuit 
$C_F$. Finally, let $C' = \maj(C_1,C_2,...,C_q)$ be the circuit obtained from $\maj(x_1,...,x_q)$ by identifying,
 for each $i\in \{1,...,q\}$, the output bit of $C_i$ with the control bit of the sub-circuit 
$\adder_i$. Then $C'$ accepts if and only if at least 
$\lceil \frac{q}{2} \rceil$ of the circuits $C_i$ accept.
For each constant $\delta$ with $0<\delta < 1$, we can choose a suitable $q=n^{O(1)}$, such that 
$C'$ accepts with probability $1$ if the $3$-SAT formula $F$ 
is satisfiable, and accepts with probability at most $\delta$ if $F$ is unsatisfiable. 
To show that $C'$ has online-width at most $w+O(\log n)$, we note that there is a topological ordering 
of the gates of $C'$ which executes all gates of $C_i$ before all gates of $\adder_i$, all gates 
of $\adder_i$ before all gates of $C_{i+1}$, and all gates of $\adder_q$ before all gates of $\comparison$ 
(Fig. \ref{figure:Majority-Vote}). $\square$

\begin{figure}[t]
\centering
\includegraphics[scale=0.65]{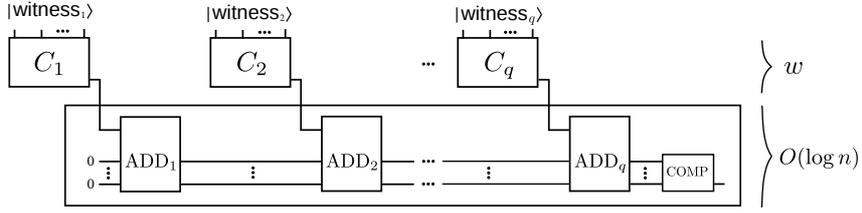}
\caption{Amplification of Probability with a logarithmic increase in online-width. The sub-circuit 
surrounded by the box implements the majority vote in online-width $O(\log n)$.} 
\label{figure:Majority-Vote}
\end{figure}

\section{Quantum Witnesses vs Quantum Verifiers of Logarithmic Width}
\label{section:ProofLogarithmicTreewidthQMA}

In this section we will prove Theorem \ref{theorem:LogarithmicTreewidthQMA}, which 
states that for any constant $\delta$ with $0<\delta< 1/2$, and any quantum circuit $C$ with 
$n$ uninitialized inputs, $\poly(n)$ gates and online-width $O(\log n)$, 
it is $\mathrm{QMA}$-Complete to determine 
whether $\mathit{Pr}^{qu}(C)\geq 1-\delta$ or whether 
${\mathit{Pr}^{qu}(C)\leq \delta}$.
We note that membership in $\mathrm{QMA}$ is trivial, since $\mathrm{QMA}$ is defined in terms of 
the quantum satisfiability of quantum circuits of polynomial online-width. On the 
other hand, the proof of $\mathrm{QMA}$-hardness will be  similar to the proof of $\mathrm{NP}$-hardness 
for the classical satisfiability of quantum circuits of logarithmic online-width given 
in Section \ref{section:ProofLogarithmicTreewidthNP}. The only difference 
is that instead of using  a reduction from 3-SAT, we will 
use a reduction from the $\mathrm{QMA}$-Complete problem $k$-local Hamiltonian \cite{KitaevShenVyalyi2002}.
We note that this completeness result implies that $\mathrm{QMA}[\onlinecutwidth,O(\log n)]=\mathrm{QMA}$.
Since, by Equation \ref{equation:ComparisonMeasures}, the treewidth of a circuit is always upper bounded 
by its online width, we also have that $\mathrm{QMA}[\treewidth,O(\log n)] = \mathrm{QMA}$. 

Let $\hilbert_d= \C^d$. 
An operator $H:\hilbert_d^{\otimes n}\rightarrow \hilbert_d^{\otimes n}$ is called a $k$-local 
Hamiltonian if it is expressible in the form $H=\sum_{j} H_j$ where each $H_j$ 
is an Hermitian operator acting on at most $k$ qubits. Additionally 
we assume a normalizing condition requiring both $H_j$ and $I-H_j$ to be positive semidefinite.

\begin{definition}[$k$-Local Hamiltonian Problem]
\label{definition:LocalHamiltonianProblem}
Let $k=O(1)$, $H=\sum_{i=1}^m H_i$ be a $k$-local Hamiltonian acting on $n$ qubits, and $a,b$ be real numbers 
such that $0\leq a < b$ and $b-a = \Omega(n^{-\alpha})$ for some constant $\alpha>0$. 
The $k$-local Hamiltonian problem consists in determining whether $H$ has an eigenvalue 
not exceeding $a$, or whether all eigenvalues of $H$ are at least $b$. 
\end{definition}

The $k$-local Hamiltonian problem was shown to be $\mathrm{QMA}$ complete in \cite{KitaevShenVyalyi2002}.
In particular, the proof that this problem is in $\mathrm{QMA}$ follows from a reduction to the problem 
of approximating the maximum acceptance probability of a quantum circuit with uninitialized inputs. 
Next, we show that this reduction carry over even if we require the obtained quantum 
circuits to have logarithmic online-width. We note that the exposition given below is similar to 
the one encountered in \cite{KitaevShenVyalyi2002}, except for some adaptations that 
take the online-width of the involved circuits into consideration.

\begin{lemma}[Adaptation from \cite{KitaevShenVyalyi2002}] 
\label{lemma:LocalHamiltonianConversion}
Let $k=O(1)$, $H=\sum_{i=1}^m H_i$ be a $k$-local Hamiltonian acting on $n$ qubits, and $a,b$ be real numbers 
such that $0\leq a < b$ and $b-a = \Omega(n^{-\alpha})$ for some constant $\alpha>0$. Then there is a quantum circuit 
$C_H$ with $n$ uninitialized inputs and online-width $O(\log n)$ satisfying the following conditions. 
\begin{enumerate}
	\item (Completeness) If some eigenvalue of $H$ is smaller than or equal to $a$, then \\
		\hphantom{spacespacespace} ${Pr^{\mathit{qu}}(C_H) \geq 1-m^{-1}a}$.
	\item (Soundness) If all eigenvalues of $H$ are at least $b$, then ${Pr^{\mathit{qu}}(C_H) \leq 1-m^{-1}b}$.
\end{enumerate}
\end{lemma}
\begin{proof}
Let $H = \sum_{j=1}^m H_j$ be a $k$-local Hamiltonian.
For each local term $H_j$ we construct a circuit implementing 
the POVM $\{H_j,I-H_j\}$. Since $H_j$ can be rewritten as $H_j=\sum_s\lambda_s \ket{\psi_s}\bra{\psi_s}$,
where $\ket{\psi_s}$ are the eigenvectors of $H_j$, and since $H_j$ acts on a constant number of qubits, 
the mentioned POVM can be implemented by a constant size circuit $W_j$ that acts on the 
qubits affected by $H_j$ and an auxiliary output qubit. The action of $W_j$ on the orthogonal 
system of eigenvectors of $H_j$ is given by 
$$W_j:\ket{\psi_s,0}\rightarrow\ket{\psi_s}\otimes (\sqrt{\lambda_s}\ket{0} + \sqrt{1-\lambda_s}\ket{1}).$$
The probability of measuring $1$ at 
the output bit of $W_j$ is given by 
$$Pr_1(W_j) =\bra{\eta,0}\,W_j^{\dagger} (I\otimes \ket{1}\bra{1})W_j\,\ket{\eta,0} = 1-\bra{\eta}H_j\ket{\eta}.$$

Now consider a circuit $C_H$ which implements the following verification process.  First, the verifier selects a number 
$r\in \{1,...,m\}$ uniformly at random. Subsequently, when reading the witness $\ket{\eta}$ provided by Merlin, the verifier 
ignores all qubits but those which are affected by $W_r$. Finally 
when all relevant qubits have been read, the verifier applies the sub-circuit $W_r$ to these relevant qubits.
The overall acceptance probability of the circuit $C_H$ is given by 
$$Pr(C_H,\ket{\eta}) = \sum_{j}\frac{1}{r}Pr(W_j,\ket{\eta}) = 1 - m^{-1} \bra{\eta} H \ket{\eta}.$$
In particular, if $\ket{\eta}$ is an eigenvector of $H$ with eigenvalue smaller than $a$, then the 
acceptance probability of $C_H$ is greater than $1-m^{-1}a$ while if every eigenvector of $H$
has eigenvalue at least $b$ then the acceptance probability of $C_H$ is at most $1-m^{-1}b$.
The circuit $C_H$ can clearly be implemented in online-width $O(\log n)$  since we just need $O(\log n)$ 
bits to implement the random choice of $r$. $\square$ 
\end{proof}

To prove Theorem \ref{theorem:LogarithmicTreewidthQMA}, it remains to show that both 
the soundness and the completeness in Lemma \ref{lemma:LocalHamiltonianConversion} 
can be amplified with only a logarithmic increase in online-width. 
Let $q = n^{O(1)}$, $p(n) =  (a(n)+b(n))/2$, and let $\maj'(x_1,...,x_q)$ be a circuit that implements the following 
variant of the majority function. 

\begin{equation}
\label{equation:majority}
\maj'(x_1,...,x_q) = \left\{
\begin{array}{lcr}
1 & & \mbox{ if $\sum_{j=1}^q x_j \geq p(n)\cdot q$} \\
0 & & \mbox{ if $\sum_{j=1}^q x_j < p(n)\cdot q$} \\
\end{array}
\right. 
\end{equation}

Let $C_1,...,C_q$ be independent copies of the circuit $C_H$. 
Let ${C' = \maj'(C_1,...,C_q)}$ be the circuit obtained from $\maj'(x_1,...,x_q)$ by 
identifying the output of $C_i$ with the $i$-th input of $\maj'(x_1,...,x_q)$. 
It can be shown (See \cite{KitaevShenVyalyi2002} Lemma $14.1$) that 
if there exists a witness $\ket{\eta}\in \hilbert_d^{\otimes n}$ such that 
$Pr(C_H,\ket{\eta})\geq a(n)$, then there is a witness $\ket{\eta'}\in \hilbert_d^{\otimes q n}$ such that 
$$Pr(C',\ket{\eta'})\geq 1-\exp(-\Omega(\mathit{poly}(n))).$$ 
On the other hand, if for every state $\ket{\eta}\in \hilbert_d^{\otimes n}$, 
$Pr(C_H,\ket{\eta})\leq b(n)$, then for every state $\ket{\eta'}\in \hilbert_d^{\otimes qn}$, 
the verifier accepts with probability at most $\exp(-poly(n))$. 
Similarly to the circuit computing $\maj(x_1,...,x_q)$ described in Section \ref{section:ProofLogarithmicTreewidthNP}, 
the circuit $\maj'(x_1,...,x_q)$ can be implemented in $O(\log_n)$ as a sequence 
$$\adder_1,\adder_2,...,\adder_q,\comparison'$$ of adder circuits followed by a comparator circuit $\comparison'$ which 
accepts if and only if the value at its input register is at least $p(n)\cdot q$. As in Section \ref{section:ProofLogarithmicTreewidthNP}, this implies that the overall circuit $C'$ has online-width at most $w+O(\log n)$, since 
we can consider an ordering of the gates of $C'$ that executes all gates of $C_i$ before the gates of the adder 
circuit $\adder_i$, all gates of $\adder_i$ before the gates of $C_{i+1}$, and all gates of $\adder_q$ before 
all gates of $\comparison'$ (Fig. \ref{figure:Majority-Vote}). $\square$

\section{Conclusion and Open Problems}
\label{section:Conclusion}

In this work we have introduced the notion of {\em feasibility tensor network}. We have shown that the problem of computing 
a classical assignment $y\in \{0,1\}^n$ that maximizes the acceptance probability of a quantum circuit $C$ with $n$ uninitialized inputs and $\poly(n)$ gates can be 
reduced to the problem of finding an initialization of maximum value for a feasibility tensor network. Using this reduction,
we have shown that if $C$ has treewidth $t$, then a $\delta$-optimal assignment for $C$ can be found in time $(n/\delta)^{\exp(O(t))}$. Therefore 
we have provided the first example of quantum optimization problem that can be solved in polynomial time on quantum circuits of constant treewidth.

We have also provided new characterizations of the complexity classes $\mathrm{NP}$ and $\mathrm{QMA}$ 
in terms of Merlin-Arthur protocols in 
which the verifier is a circuit of logarithmic treewidth, by showing that $\mathrm{QCMA}[\treewidth,O(\log n)] = \mathrm{NP}$ and 
that $\mathrm{QMA}[\treewidth,O(\log n)] = \mathrm{QMA}$. In other words, we have shown 
that quantum witnesses are inherently more powerful than classical witnesses for Merlin-Arthur protocols 
with verifiers of logarithmic treewidth, assuming $\mathrm{QMA}\neq \mathrm{NP}$. Our main theorem 
implies that $\mathrm{QCMA}[\treewidth,O(1)] \subseteq \mathrm{P}$. However we were not able to determine 
whether an analog inclusion can be proved 
when the verifier has constant width and the witness is allowed to be an arbitrary quantum state. More precisely,
the following question is left open: Is $\mathrm{QMA}[\treewidth,O(1)]\subseteq \mathrm{P}$?

The $\mathrm{NP}$-hardness of the problem of computing optimal classical assignments for quantum circuits of logarithmic treewidth 
imposes some constraints on the possibility of drastically improving the running time of our algorithm. 
However we leave the following question open: Is the problem of computing $\delta$-optimal classical 
assignments for quantum circuits in FPT with respect to treewidth? More precisely, can this problem be solved in time $f(t)\cdot \poly(n,\delta)$? We observe 
that while in the case of classical circuits one can determine the existence of a satisfying assignment in time $2^{O(t)}\cdot n^{O(1)}$ 
\cite{AlekhnovichRazborov2002,AllenderChenLouPeriklisPapakonstantinouTang2014}, 
the fact that $\mathrm{QCMA}[\treewidth,O(\log n)] = \mathrm{NP}$ 
implies that in the case of quantum circuits the function $f(t)$ should be at 
least double exponential in $t$, assuming the exponential time 
hypothesis (ETH) \cite{ImpagliazzoPaturi2001}.

\bibliographystyle{abbrv}

\begin{thebibliography}{10}

\bibitem{AharonovKitaevNisan1998}
D.~Aharonov, A.~Kitaev, and N.~Nisan.
\newblock Quantum circuits with mixed states.
\newblock In {\em Proc. of the 30th Symposium on Theory of Computing}, pages
  20--30, 1998.

\bibitem{AharonovNaveh2002}
D.~Aharonov and T.~Naveh.
\newblock Quantum {NP} - {A} survey.
\newblock {\em arXiv preprint quant-ph/0210077}, 2002.

\bibitem{AlekhnovichRazborov2002}
M.~Alekhnovich and A.~A. Razborov.
\newblock Satisfiability, branch-width and Tseitin tautologies.
\newblock In {\em Proc. of the 43rd Symposium on Foundations of Computer
  Science}, pages 593--603, 2002.

\bibitem{AllenderChenLouPeriklisPapakonstantinouTang2014}
E.~Allender, S.~Chen, T.~Lou, P.~A. Papakonstantinou, and B.~Tang.
\newblock Width-parametrized {SAT}: Time--space tradeoffs.
\newblock {\em Theory of Computing}, 10(12):297--339, 2014.

\bibitem{ArnborgLagergrenSeese1991}
S.~Arnborg, J.~Lagergren, and D.~Seese.
\newblock Easy problems for tree-decomposable graphs.
\newblock {\em Journal of Algorithms}, 12(2):308--340, 1991.

\bibitem{ArnborgProskurowski1989}
S.~Arnborg and A.~Proskurowski.
\newblock Linear time algorithms for {NP}-hard problems restricted to partial
  $k$-trees.
\newblock {\em Discrete Applied Mathematics}, 23(1):11--24, 1989.

\bibitem{Babai1992}
L.~Babai.
\newblock Bounded round interactive proofs in finite groups.
\newblock {\em SIAM Journal on Discrete Mathematics}, 5(1):88--111, 1992.

\bibitem{Bodlaender1988}
H.~L. Bodlaender.
\newblock Classes of graphs with with bounded treewidth.
\newblock Bulletin of the EATCS, 36:116-126, 1988.

\bibitem{Bodlaender1989}
H.~L. Bodlaender.
\newblock {NC}-algorithms for graphs with small treewidth.
\newblock In {\em Proc. of the 14th International Workshop on Graph-Theoretic
  Concepts in Computer Science}, volume 344 of {\em LNCS}, pages 1--10, Springer, 1989.

\bibitem{BodlaenderFominKosterKratschThilikos2012}
H.~L. Bodlaender, F.~V. Fomin, A.~M. Koster, D.~Kratsch, and D.~M. Thilikos.
\newblock On exact algorithms for treewidth.
\newblock {\em ACM Transactions on Algorithms}, 9(1):12, 2012.

\bibitem{Bookatz2014}
A.~D. Bookatz.
\newblock {QMA}-complete problems.
\newblock {\em Quantum Information {\&} Computation}, 14(5-6):361--383, 2014.

\bibitem{BroeringLokamSatyanarayana2004}
E.~Broering and S.~V. Lokam.
\newblock Width-based algorithms for {SAT} and {CIRCUIT-SAT}.
\newblock In {\em Proc. of the 6th International Conference on Theory and Applications of Satisfiability Testing}, 
volume 2919 of {\em LNCS}, pages  162--171. Springer, 2004.

\bibitem{Courcelle1990Monadic}
B.~Courcelle.
\newblock The monadic second-order logic of graphs {I}. {R}ecognizable sets of
  finite graphs.
\newblock {\em Information and computation}, 85(1):12--75, 1990.

\bibitem{deOliveiraOliveira2015Satisfiability}
M. de Oliveira Oliveira.
\newblock On the satisfiability of quantum circuits of small treewidth. 
\newblock In {\em Proc. of the 10th International Computer Science Symposium in Russia}, volume 9139 of {\em LNCS}, 
pages 157--172, Springer, 2015.

\bibitem{GeorgiouKonstantinosPapakonstantinou2008}
K.~Georgiou and P.~A. Papakonstantinou.
\newblock Complexity and algorithms for well-structured k-SAT instances.
\newblock In {\em Proc. of the 11th International Conference on Theory and
  Applications of Satisfiability Testing}, volume 4996 of {\em LNCS}, pages 105--118, Springer, 2008.

\bibitem{Gottesman1998}
D.~Gottesman.
\newblock The {H}eisenberg representation of quantum computers.
\newblock {\em arXiv preprint quant-ph/9807006}, 1998.

\bibitem{ImpagliazzoPaturi2001}
R.~Impagliazzo and R.~Paturi.
\newblock On the complexity of $k$-{SAT}.
\newblock {\em Journal of Computer and System Sciences} 62(2):367--375, 2001.


\bibitem{JozsaLinden2003}
R.~Jozsa and N.~Linden.
\newblock On the role of entanglement in quantum-computational speed-up.
\newblock {\em Proc. of the Royal Society of London, Series A},
  459(2036):2011--2032, 2003.

\bibitem{KitaevShenVyalyi2002}
A.~Kitaev, A.~Shen, and M.~Vyalyi.
\newblock {\em Classical and Quantum Computation}, volume~47 of {\em Graduate
  Studies in Mathematics}.
\newblock {AMS}, 2002.

\bibitem{MarkovShi2008}
I.~L. Markov and Y.~Shi.
\newblock Simulating quantum computation by contracting tensor networks.
\newblock {\em SIAM Journal on Computing}, 38(3):963--981, 2008.

\bibitem{NielsenChuang2010}
M.~A. Nielsen and I.~L. Chuang.
\newblock {\em Quantum computation and quantum information}.
\newblock Cambridge university press, 2010.

\bibitem{RobertsonSeymour1984}
N.~Robertson and P.~D. Seymour.
\newblock Graph minors {III}. {P}lanar tree-width.
\newblock {\em Journal of Combinatorial Theory, Series B}, 36(1):49--64, 1984.

\bibitem{RobertsonSeymour1995}
N.~Robertson and P.~D. Seymour.
\newblock Graph minors {XIII}. {T}he disjoint paths problem.
\newblock {\em Journal of Combinatorial Theory, Series B}, 63(1):65--110, 1995.

\bibitem{BodlaenderThilikosSerna2000}
D.~M. Thilikos, M.~J. Serna, and H.~L. Bodlaender.
\newblock Constructive linear time algorithms for small cutwidth and
  carving-width.
\newblock In {\em Proc. of the 11th International Conference on Algorithms and
  Computation}, volume 1969 of {\em LNCS}, pages 192--203, Springer, 2000.

\bibitem{Valiant2002}
L.~G. Valiant.
\newblock Quantum circuits that can be simulated classically in polynomial
  time.
\newblock {\em SIAM Journal on Computing}, 31(4):1229--1254, 2002.

\bibitem{Vidal2003}
G.~Vidal.
\newblock Efficient classical simulation of slightly entangled quantum
  computations.
\newblock {\em Physical Review Letters}, 91:147902, Oct 2003.

\bibitem{Watrous2000}
J.~Watrous.
\newblock Succinct quantum proofs for properties of finite groups.
\newblock In {\em Proc. of the 41st Symposium on Foundations of Computer
  Science}, pages 537--546, 2000.

\end{thebibliography}

\end{document}